\documentclass[10pt]{article}

\usepackage{latexsym,amsmath,amsthm}
\usepackage{amsfonts}
\usepackage{amssymb}
\usepackage{latexsym}

\usepackage[T1]{fontenc}
\usepackage{fourier}
\usepackage{bbm}

\newcommand\pa[1]{#1}
\newcommand\cg[1]{}

\newcommand\good{tame}

\newcommand\Dtame{\mathcal{D}_{\text{tame}}}
\newcommand\dtame[1]{\mathcal{D}_{{#1},\text{tame}}}

\newcommand\myconst{(1.01/\unif)^{1/(\unif-1)}}
\newcommand\hyp{\cG}
\newcommand\fhyp{G}
\newcommand\COLS{B(\hyp)}
\newcommand\COLSH{B(\fhyp)}

\newcommand\clusterH{\cC({\fhyp,\sigma)}}
\newcommand\cols{q}
\newcommand\unif{k}
\newcommand\hnc{\hyp(n,\unif,cn)}
\newcommand\plantedhyp{\hyp(n,\unif,cn,\sigma)}
\newcommand\phyp{\hyp_\sigma}

\newcommand\gnp{G(n,p)}

\numberwithin{equation}{section}

\def\vec#1{\mathchoice{\mbox{\boldmath$\displaystyle#1$}}
{\mbox{\boldmath$\textstyle#1$}}
{\mbox{\boldmath$\scriptstyle#1$}}
{\mbox{\boldmath$\scriptscriptstyle#1$}}}

\newcommand{\Zkc}{Z_{\cols}}
\newcommand{\Zkb}{Z_{\cols,\mathrm{bal}}}

\newcommand{\Zkg}{Z_{\cols,\mathrm{\good}}}

\newcommand\COMB{Combinatorica}

\DeclareMathOperator{\pr}{\mathbf P}

 \usepackage[colorlinks=true,citecolor=black,linkcolor=black,urlcolor=blue]{hyperref}

\newtheorem{definition}{Definition}[section]
\newtheorem{claim}[definition]{Claim}

\newtheorem{theorem}[definition]{Theorem}
\newtheorem{lemma}[definition]{Lemma}
\newtheorem{proposition}[definition]{Proposition}
\newtheorem{corollary}[definition]{Corollary}

\newtheorem{fact}[definition]{Fact}

\usepackage[dvips]{epsfig}
\graphicspath{{./Fig_eps/}{./Eps/}}
\DeclareGraphicsExtensions{.ps,.eps}
\usepackage{color}
\usepackage{fullpage}

\newcommand\id{\mathrm{id}}

\newcommand\cA{\mathcal{A}}

\newcommand\cC{\mathcal{C}}
\newcommand\cD{\mathcal{D}}

\newcommand\cG{\mathcal{G}}
\newcommand\cE{\mathcal{E}}

\newcommand\cS{\mathcal{S}}
\newcommand\cT{\mathcal{T}}

\def\cR{{\mathcal R}}
\def\cC{{\mathcal C}}
\def\cE{{\mathcal E}}

\newcommand\eps{\varepsilon}

\newcommand\Erw{\mathbf{E}}

\newcommand{\vecone}{\vec{1}}

\newcommand\bc[1]{\left({#1}\right)}
\newcommand\cbc[1]{\left\{{#1}\right\}}

\newcommand\brk[1]{\left\lbrack{#1}\right\rbrack}

\newcommand\norm[1]{\left\|{#1}\right\|}
\newcommand\abs[1]{\left|{#1}\right|}

\newcommand{\whp}{w.h.p.}

\newcommand{\Erdos}{Erd\H{o}s}
\newcommand{\Renyi}{R\'enyi}

\newcommand{\Bollobas}{Bollob\'as}

\newcommand{\Luczak}{\L uczak}

\newcommand\Def{Definition}
\newcommand\Lem{Lemma}
\newcommand\Prop{Proposition}
\newcommand\Thm{Theorem}
\newcommand\Cor{Corollary}
\newcommand\Sec{Section}

\allowdisplaybreaks

\begin{document}

\title{{Hypergraph coloring up to condensation}
	}

\author{Peter Ayre\\
\small School of Mathematics and Statistics\\[-0.5ex]
\small UNSW Sydney\\[-0.5ex]
\small NSW 2052, Australia\\
\small \texttt{peter.ayre@unsw.edu.au}\\
\and Amin Coja-Oghlan\thanks{The research leading to these results has received funding from the European Research Council under the European Union's Seventh Framework Programme (FP/2007-2013) / ERC Grant Agreement n.\ 278857--PTCC.}\\
\small Mathematics Institute\\[-0.5ex]
\small Goethe University \\[-0.5ex]
\small Frankfurt 60325, Germany\\
\small \texttt{acoghlan@math.uni-frankfurt.de} \\
\and Catherine Greenhill\thanks{Research supported by the Australian Research Council Discovery Project DP140101519.}\\
\small School of Mathematics and Statistics\\[-0.5ex]
\small UNSW Sydney\\[-0.5ex]
\small NSW 2052, Australia\\
\small \texttt{c.greenhill@unsw.edu.au}
}

\date{April 13, 2018}

\maketitle

\begin{abstract}
\noindent
Improving a result of Dyer, Frieze and Greenhill [Journal of Combinatorial Theory, Series B, 2015],
we determine the $\cols$-colorability threshold in random $\unif$-uniform hypergraphs up to an additive error of
$\ln2+\eps_\cols$, where $\lim_{\cols\to\infty}\eps_\cols=0$.
The new lower bound on the threshold matches the ``condensation phase transition'' predicted by statistical physics considerations
	[Krzakala et al., PNAS~2007].

\end{abstract}

\section{Introduction}

Recent work on random constraint satisfaction problems has focused either on the case of binary variables and $k$-ary constraints ({e.g.,} random $k$-SAT)
or on the case of $k$-ary variables and binary constraints (e.g., random graph coloring) for some $k\geq3$.
In these two cases substantial progress has been made over the past few years.
For instance, the $k$-SAT threshold has been identified precisely for large enough $k$~\cite{DSS3}.
Moreover, in the random hypergraph $2$-coloring problem (or equivalently the $k$-NAESAT problem) the threshold is known up to an error term
that tends to $0$ rapidly in terms of the size $k$ of the edges~\cite{KostaNAE}.
In addition, the best current upper and lower bounds on the $k$-colorability threshold of the \Erdos-\Renyi\ random graph are within a small {additive} constant~\cite{Danny}.
By comparison, little is known about problems in which both the arity of the constraints and the domain of the variables have size greater than two.
Although it has been asserted that the techniques developed in recent work should carry over~\cite{Danny}, this claim has hardly been put to the test.

The present paper deals with one of the most natural examples of a problem with $\unif$-ary constraints and $\cols$-ary variables with $\cols,\unif\geq3$,
namely $\cols$-colorability of random $\unif$-uniform hypergraphs.
% For B1
Let $[m]$ denote the set $\{1,\ldots, m\}$ for any positive integer $m$.
To be precise, by a $\cols$-coloring of $\fhyp=(V,E)$ we mean a map $\sigma:V\to[\cols]$ such that $|\sigma(e)|>1$ for all $e\in E$, i.e., no edge is monochromatic.
The chromatic number of {$\fhyp$} is the least $\cols$ for which a $\cols$-coloring exists.
The random hypergraph model that we consider is the most natural one, i.e., $\hyp\in\hyp(n,\unif,m)$ is a (simple) $\unif$-uniform hypergraph
on the vertex set $[n]:=\{1,2,3,\dots,n\}$ with a set of precisely $m$ edges chosen uniformly at random. 

For every $\cols\geq2,\unif\geq3$ there exists a (non-uniform) sharp threshold $c_{\cols,\unif}=c_{\cols,\unif}(n)$ for $\cols$-colorability~\cite{HatamiMolloy}.
That is, if $m=m(n)$ is a sequence such that for some fixed $\eps>0$ we have $m(n)<(1-\eps)nc_{\cols,\unif}(n)$, then $\hyp(n,\unif,m)$ is $\cols$-colorable \whp,
whereas {\whp}\ the random hypergraph fails to be $\cols$-colorable if $m(n)>(1+\eps)nc_{\cols,\unif}(n)$.
The best prior bounds on this threshold, obtained by Dyer, Frieze and Greenhill~{\cite[Remark~2.1,~(82)]{DFG}}, read
	\begin{align}\label{eqDFG}
	(\cols^{\unif-1}-1)\ln \cols-1-\eps_{\cols,\unif}
	\leq\liminf_{n\to\infty}\, c_{\cols,\unif}(n)\leq\limsup_{n\to\infty}c_{\cols,\unif}(n)\leq\bc{\cols^{\unif-1}-{1/2}}\ln \cols,
	\end{align}
where $\lim_{\cols\to\infty}\eps_{\cols,\unif}=0$ for any fixed $\unif\geq3$.
Thus, the upper and the lower bound differ by an additive $\frac12\ln \cols+1+\eps_{\cols,\unif}$, a term that diverges in the limit of large $\cols$.
The main result of this paper provides an improved lower bound that is 
within an additive $\ln2$ 
of the upper bound from~(\ref{eqDFG}),
{in the large-$q$ limit}.

\begin{theorem}\label{Thm_main}
For each $\unif\geq3$ {there is a number $q_0=q_0(k)>0$ such that for all $q>q_0$ we have}
	\begin{align*}
	\liminf_{n\to\infty} \, c_{\cols,\unif}(n)&\geq(\cols^{\unif-1}-1/2)\ln \cols-\ln2-{1.01\ln \cols/\cols}.
	\end{align*}
\end{theorem}

The proof of \Thm~\ref{Thm_main} is based on the second moment method.
So is~\cite{DFG}, which generalises the second moment argument of Achlioptas and Naor~\cite{AchNaor} from graphs to hypergraphs.
The result of Achlioptas and Naor was recently improved by Coja-Oghlan and Vilenchik~\cite{Danny},
and in this paper we generalise the argument from that paper to hypergraphs.
While numerous details need adjusting, the basic proof strategy that we pursue is similar to the one suggested in~\cite{Danny}.
In particular, the improvement over~\cite{DFG} results from studying the second moment of a subtly chosen random variable.
While the random variable considered in~\cite{DFG} is just the number of (balanced) $\cols$-colorings of the random hypergraph,
here we use a random variable that is inspired by ideas from statistical mechanics; we will give a more detailed outline in \Sec~\ref{sec:outline} below.
Thus, the present paper shows that, indeed, with a fair number of careful modifications the method from~\cite{Danny} can be generalised to hypergraphs.

% Moved this section here, for C8. OK?
\paragraph{Notation.}\label{Apx_notation}
We assume throughout that the number of vertices, $n$, is sufficiently
large for our estimates to hold. We also assume that the number of
colors, $\cols$, exceeds some large enough constant $\cols_0=\cols_0(\unif)$.
But of course $\cols,\unif$ are always assumed to remained fixed as $n\to\infty$.

We use the $O$-notation to refer to the limit $n\rightarrow\infty$. 
For example, $f(n)=O(g(n))$ means that there exists some $C>0, n_0>0$ such that for all $n>n_0$ we have $|f(n)|\leq C\cdot|g(n)|$. 
In addition, $o(\cdot), \Omega(\cdot),\Theta(\cdot)$ take their usual definitions, {except that we assume the expression $\Omega(n)$ is positive (for
sufficiently large $n$) whenever we write $\exp(-\Omega(n))$}. We write $f(n)\sim g(n)$ if $\lim_{n\rightarrow\infty}f(n)/g(n)=1$.

When discussing estimates that hold in the limit of large $\cols$ we will make 
this explicit by adding the subscript $\cols$ to the asymptotic notation. 
Therefore, $f(\cols)=O_\cols(g(\cols))$ means that 
% Reworded for C9
there exists positive constants $C$, $q_0$ such that for all $\cols>\cols_0$ we have $|f(\cols)|\leq C\cdot |g(\cols)|$. 
Furthermore, we will write $f(\cols)=\widetilde{O}_\cols(g(\cols))$ to indicate that there exists positive $C$, $\cols_0$ such that for all $\cols>\cols_0$ we have $|f(\cols)|\leq (\ln \cols)^C\cdot |g(\cols)|$.

\section{Related work}

The quest for the chromatic number of random graphs (i.e., $\hyp(n,2,m)$) goes back to the seminal 1960 paper of \Erdos\ and \Renyi\ in which they established
the ``giant component'' phase transition~\cite{ER}.
But it took almost thirty years until a celebrated paper of \Bollobas~\cite{BBColor} determined the asymptotic value of the chromatic number of dense random graphs.
His proof used martingale tail bounds, which were introduced to combinatorics by Shamir and Spencer~\cite{ShamirSpencer} to investigate
the concentration of the chromatic number.
Building upon ideas of Matula~\cite{Matula}, \Luczak~\cite{LuczakColor} determined the asymptotic value of the chromatic number
of the \Erdos-\Renyi\ random graph in the case that $m=m(n)$ satisfies $m/n\to\infty$.
However, the results from~\cite{BBColor,LuczakColor} only determine the chromatic number up to a {\em multiplicative} error of $1+o(1)$ as $n\to\infty$,
and the resulting error term exceeds the width within which the chromatic number is known to be concentrated.
Indeed, in the case that $m=m(n)\leq {n^{3/2-\Omega(1)}}$ it is known that the chromatic number of the random graph is concentrated on two subsequent integers~\cite{AlonKriv,Luczak}.
In the sparse case $m=O(n)$ the precise values of these two integers are implied by the current bounds on the
$\cols$-colorability threshold~\cite{AchNaor,Covers,Danny}.

The $2$-colorability problem in random hypergraphs, which is essentially equivalent to the random $k$-NAESAT problem, has also been studied.
Achlioptas and Moore~\cite{nae,AchMooreHyp2} showed that the $2$-colorability threshold can be approximated within a small additive constant
via the second moment method.
Furthermore, Coja-Oghlan and Zdeborov\'a~\cite{Lenka} established the existence of a further phase transition apart from the threshold for $2$-colorability,
	the ``condensation phase transition''.
The name derives from an intriguing connection to the statistical mechanics of glasses~\cite{Kauzmann48,pnas}.
Moreover, {the argument of }Coja-Oghlan and Panagiotou~\cite{KostaNAE} determines the $2$-colorability threshold in $k$-uniform random hypergraphs up to an additive error
term $\eps_k$ that tends to $0$ exponentially as a function of $k$.

Prior to the aforementioned work of Dyer, Frieze and Greenhill~\cite{DFG} the $\cols$-colorability problem in hypergraphs
was studied by Krivelevich and Sudakov~\cite{KrivelevichSudakov}, who also considered other possible notions of colorings.
Their results are of a similar nature to \Luczak's~\cite{LuczakColor} in the case of graphs.
That is, they determine the value of the chromatic number up to a multiplicative $1+o(1)$ {factor}, with $o(1)$ hiding a term that vanishes as $m/n\to\infty$.
The same is true of the results of Kupavskii and Shabanov~\cite{KS}, which partly improve upon~\cite{KrivelevichSudakov}.
{However}, the bounds on the $\cols$-colorability threshold that can be read out of~\cite{KrivelevichSudakov,KS} are less precise than those obtained in~\cite{DFG}
	(upon which \Thm~\ref{Thm_main} improves).

\section{Outline}\label{sec:outline}

{Throughout, we assume that $n$ is sufficiently large for our error
estimates to hold, and that $q  > q_0$. Further, we assume that {$m=\lceil cn\rceil$} and for ease of notation will often write $cn$ rather than $\lceil cn\rceil$.}
\paragraph{The second moment method.}
The second moment method has become the mainstay for lower-bounding satisfiability thresholds~\cite{nae,ANP,FriezeWormald}.

Suppose that we can construct a non-negative random variable {$Z$} 
{on $\hyp(n,k,cn)$} such that the event $Z(\fhyp)>0$ implies 
$\cols$-colorability, and such that 
	\begin{align}\label{eqsmm}
	\Erw[Z^2]=O(\Erw[Z]^2)\qquad\mbox{as }n\to\infty.
	\end{align}
Then the Paley-Zygmund inequality implies that
	\begin{align}\label{eqPZ}
	\liminf_{n\to\infty} \, \pr\brk{Z>0}&\geq\liminf_{n\to\infty} \, \frac{\Erw[Z]^2}{\Erw[Z^2]}>0.
	\end{align}
{Combining (\ref{eqPZ})}
with the sharp threshold result from~\cite{HatamiMolloy}, which establishes the existence of a sharp
threshold sequence $c_{\cols,\unif}(n)$,  
{yields}
	$\liminf_{n\to\infty} \, c_{\cols,\unif}(n)\geq c$.
Hence, the second moment method can be summarised as follows.

\begin{fact}\label{Fact_smm}
% For C1
If there is a  {non-negative} random variable $Z$ {on $\hyp(n,k,cn)$}
such that $Z > 0$ implies $\cols$-colorability and (\ref{eqsmm}) holds, then
 $\liminf_{n\to\infty} \, c_{\cols,\unif}(n)\geq c$. 
\end{fact}

\noindent
Thus, our task is to exhibit a random variable $Z$ {on $\hyp(n,k,cn)$} that satisfies (\ref{eqsmm}) for as large a value of $c$ as possible.

\paragraph{Balanced colorings.}
Certainly the most natural choice for $Z$ seems to be the number $\Zkc$ of $\cols$-colorings of the random hypergraph.
Clearly, $\Zkc\geq0$ and {$\Zkc(\fhyp)>0$} only if {$\fhyp$} is $\cols$-colorable.
However, technically $\Zkc$ is a bit unwieldy.
Therefore, following Achlioptas and Naor~\cite{AchNaor}, Dyer, Frieze and Greenhill~\cite{DFG} considered a 
slightly modified random variable.
Namely, let us call a map $\sigma:[n]\to[\cols]$ {\em balanced} if {$|\sigma^{-1}(j)-n/\cols|\leq \sqrt{n}$} for all 
$j\in[\cols]$ and let $\Zkb$ be the number of balanced $\cols$-colorings of $\hyp$.

\begin{lemma}[{\cite{DFG}}]\label{Lemma_DFGfirst}
For any $\cols,\unif\geq3$ and any $c>0$ we have
% See A5.  
\begin{equation}
\label{perfect}  \Erw[\Zkb] = %\Omega\Big(n^{(1-q^2)/2}\Big)\,
		\Theta\brk{\left( q\big(1-q^{1-k}\big)^c\right)^n}
\end{equation}
and
\begin{align}\label{eqLemma_DFGfirst}
\lim_{n\to\infty}\frac1n\ln\Erw[\Zkc]=
	\lim_{n\to\infty}\frac1n\ln\Erw[\Zkb] 
=\ln\cols+c\ln\bc{1-\cols^{1-\unif}}.
\end{align}
\end{lemma}

\begin{proof}
Calculations similar to the following ones were performed in~\cite{DFG}; we repeat them here to keep the paper self-contained.
Given a balanced map $\sigma:[n]\to[q]$, let $\alpha_i=|\sigma^{-1}(i)|/n$ for $i\in [q]$ and define
$\alpha=(\alpha_1,\ldots,\alpha_q)$.
Stirling's formula yields
	\begin{align}\label{eqFirstMmt1}
	\pr\brk{\sigma\mbox{ is a proper $q$-coloring of $\hyp$}}&=
	\binom{\binom{n}{k}- \sum_{i=1}^q \binom{\alpha_in}{k}}{cn}\, \binom{\binom{n}{k}}{cn}^{-1}
		=\Theta\bc{\exp\brk{cn\, \ln\bc{1-\sum_{i\in [q]}\alpha_i^k}}},
	\end{align}
cf.\ \cite[equation (8)]{DFG}.
Let $\bar\alpha=(1/q,\ldots,1/q)$ denote the uniform distribution on $[q]$.
%such that $\TV{\alpha-\bar\alpha}=O(n^{-1/2})$.
The gradient of the function 
%	\begin{equation}\label{eqFirstMmtf}
	$f:(x_1,\ldots,x_k)\mapsto\ln\bc{1-\sum_{i\in [q]} x_i^k}$
%	\end{equation}
at the point $\bar\alpha$ is simply the vector 
	$\nabla f(\bar\alpha)$ with every entry equal to $k(1-q^{k-1})^{-1}$.
Consequently, because $\sum_{i\in [q]} (\alpha_i-1/q)=0$, expanding $f$ to the second order around $\bar\alpha$ yields
	\begin{align}\label{eqFirstMmt2}
	\ln\bc{1-\sum_{i\in [q]}} \alpha_i^k &=f(\alpha)=f(\bar\alpha)+\nabla f(\bar\alpha)(\alpha-\bar\alpha)+O(\|\alpha-\bar\alpha\|_2^2)
		=\ln(1-q^{1-k})+O(\|\alpha-\bar\alpha\|_2^2).
	\end{align}
Since $\sigma$ is balanced, we have the bound $\|\alpha-\bar\alpha\|_2^2=O(1/n)$.
Therefore, combining (\ref{eqFirstMmt1}) and (\ref{eqFirstMmt2}), we obtain
	\begin{align}\label{eqFirstMmt3}
	\pr\brk{\sigma\mbox{ is a proper $q$-coloring of $\hyp$}}&=\Theta{\bc{(1-q^{1-k})^{cn}}},
	\end{align}
uniformly for all balanced $\sigma$.
Finally, the number of balanced maps corresponding to a given $\alpha$ is 
$\binom{n}{\alpha_1 n,\ldots, \alpha_q n} = \Theta(n^{(1-q)/2})\, q^n$, by Stirling's formula, 
and the number of choices for the vector $\alpha$ is $\Theta(n^{(q-1)/2})$.
Hence the total number of balanced maps $\sigma$ is $\Theta(\cols^n)$.
Combining this with (\ref{eqFirstMmt3}) implies (\ref{perfect}).

Next, as observed in the proof of~\cite[Lemma 2.1]{DFG}, the probability that a map $\sigma:[n]\to [q]$
is a $q$-colouring of  $\hyp$ is maximised when $\sigma$ is perfectly balanced,
and this probability equals $O(1)\, \big(1 - q^{1-k})^{cn}$. (Here the $O(1)$ factor is
needed only when $q$ does not divide $n$.)
Hence, by linearity of expectation, 
\[ \Erw[\Zkb]  \leq \Erw[\Zkc] = O(1)\, \left( q(1-q^{1-k})^c\right)^n,\]
which differs from (\ref{perfect}) by at most a constant factor. This implies (\ref{eqLemma_DFGfirst}),
completing the proof.
\end{proof}

It is easily verified that the r.h.s.\ of (\ref{eqLemma_DFGfirst}) is positive if $c<(\cols^{\unif-1}-\frac12)\ln\cols-\ln2$.
Hence, for such $c$, both $\Erw[\Zkc]$ and  $\Erw[\Zkb]$ are exponential in $n$.
They differ only in their sub-exponential terms.
Consequently, we do not give anything away by confining ourselves to balanced colorings only.
In the following we will see why neither $\Zkc$ nor $\Zkb$ is a good random variable to work with and why neither can be used to prove \Thm~\ref{Thm_main}.
What we learn will guide us towards constructing a better random variable.

While working out the first moment of $\Zkb$ (i.e., the proof of \Lem~\ref{Lemma_DFGfirst}) is pretty straightforward,
getting a handle on the second moment is not quite so easy.
% For C2: inserted "balanced"
Of course, the second moment of $\Zkb$ is nothing but the expected number 
of {\em pairs} of balanced $\cols$-colorings.
Moreover, the probability that two maps $\sigma,\tau:\brk n\to\brk\cols$ simultaneously happen to be $\cols$-colorings of $\hyp$
will depend on how ``similar'' $\sigma,\tau$ are.
{To gauge similarity, define the} {\em overlap} of $\sigma,\tau$ as the $\cols\times\cols$-matrix $a(\sigma,\tau)=(a_{ij}(\sigma,\tau))_{i,j\in[\cols]}$ with entries
	\begin{align*}
	a_{ij}(\sigma,\tau)&=n^{-1}|\sigma^{-1}(i)\cap\tau^{-1}(j)|.
	\end{align*}
In words, $a_{ij}(\sigma,\tau)$ is the probability that a random vertex $v\in[n]$ has color $i$ under $\sigma$ and color $j$ under $\tau$.
Then we can cast the second moment in terms of the overlap as follows.
Let $\cR=\cR_{n,\cols}$ be the set of all overlaps $a(\sigma,\tau)$ of balanced $\sigma,\tau:\brk n\to\brk\cols$.
Though the results of the next lemma can be found in~\cite{DFG}, for completeness we provide a brief proof here.

\begin{lemma}[{\cite{DFG}}]\label{Lemma_DFGsecond}
Let $\norm a_{\unif}=\brk{\sum_{i,j\in[\cols]}a_{ij}^{\unif}}^{1/{\unif}}$ be the $\ell_\unif$-norm and define
	\begin{align*}
	H(a)&=-\sum_{i,j\in[\cols]}a_{ij}\ln a_{ij},&
	E(a)&=E_{\cols,c,\unif}(a)=c\ln\brk{1-2\cols^{1-\unif}+\norm a_{\unif}^{\unif}}.
	\end{align*}
Let $F(a)=H(a)+E(a)$ and suppose that $a^\ast\in \mathcal{R}$ satisfies $F(a^\ast) = \max_{a\in\mathcal{R}} F(a)$.
%% We had a typo here before: should be (1-q^2)/2,  not (1-k^2)/2.  Fixed now.
Then
	\begin{equation}\label{eqLemma_DFGsecond}
  \Erw[\Zkb^2] = \exp\brk{n F(a^\ast) + o(n)}.
	\end{equation}
Next, let 
$\xi$ be a positive constant and suppose that $\mathcal{A}\subseteq \mathcal{R}$ has the
following property: $a_{ij}\geq \xi $ for all $a\in \mathcal{A}$ and all $i,j\in [q]$. 
Then
\begin{equation}
\label{eqLemma_DFGsecond_precise}
	\Erw[\Zkb^2\cdot \textbf{1}_{\mathcal{A}} ] = \Theta(n^{(1-q^2)/2})\, \sum_{a\in\mathcal{A}} \exp\brk{n F(a)}.
\end{equation}
\end{lemma}

\begin{proof}
First, observe that for a given $a\in\cR$, the number of $\sigma,\tau$ with overlap $a$ is given by 
the multinomial coefficient
\[\binom{n}{ a_{11}n, a_{12}n,\ldots, a_{qq}n}.
\] 
Next, fix balanced maps $\sigma,\tau$ with overlap $a$. 
By inclusion-exclusion,
the probability that a random edge chosen uniformly out of all 
$\binom{n}{\unif}$ possible edges is monochromatic under either $\sigma$ or $\tau$ 
equals
	\begin{align}
\label{eqFirstMmt10}
	\binom nk^{-1}&\brk{\sum_{i\in [q]}\,
   \brk{\binom{\sum_{j\in [q]} a_{ij}n}k+\binom{\sum_{j\in [q]}a_{ji}n}k}-\sum_{i,j\in [q]}\binom{a_{ij}n}k} \nonumber \\
		&=\sum_{i\in [q]}\brk{\bc{\sum_{j\in [q]} a_{ij}}^k+\bc{\sum_{j\in [q]} a_{ji}}^k}-\sum_{i,j\in [q]}a_{ij}^k+O\bc{\frac1n}.
	\end{align}
To simplify this expression we observe that since $\sigma,\tau$ are balanced,
	\begin{align*}
	\sum_{i\in [q]}\bc{\sum_{j\in [q]} a_{ij}}^k
   &=\sum_{i\in [q]}\bc{q^{-1}-\left(q^{-1}-\sum_{j\in [q]} a_{ij}\right)}^k\\
		&= q^{1-k}-kq^{1-k}\bc{1-\sum_{i,j\in [q]} a_{ij}}
			+O\bc{\sum_{i\in [q]}\bc{\frac1q-
   \sum_{j\in [q]} a_{ij}}^2}\\
		&=q^{1-k}+O(1/n),
	\end{align*}
because $\sum_{i,j\in [q]}a_{ij}=1$ and 
$|q^{-1}-\sum_{j\in [q]} a_{ij}|=O(n^{-1/2})$.
Of course, the same steps apply to $\sum_{i\in [q]}\bc{\sum_{j\in [q]} a_{ji}}^k$.
Hence, since $\sigma$ and $\tau$ are balanced, (\ref{eqFirstMmt10})
can be written as
	\begin{align*}
		2q^{1-k}-\sum_{i,j=1}^qa_{ij}^k+O\bc{\frac1n}.
	\end{align*}
Therefore
\begin{equation}
\label{Esum}
 \Erw[\Zkb^2] \sim \sum_{a\in\mathcal{R}}\, \binom{n}{a_{11}n,\, a_{12}n,\ldots, a_{qq}n}\,
   \exp(n E(a) + O(1)).
\end{equation}
Let $b\vee 1$ denote $\max\{ b,1\}$. We give upper and lower bounds on the multinomial
coefficient by applying Stirling's formula in the form 
\[ b! = \sqrt{2\pi (b\vee 1)}\, \left(\frac{b}{e}\right)^b\, 
   \left[ 1 + O\left(\frac{1}{b+1}\right)\right],\]
which holds for all nonnegative integers $b$.  This gives
\begin{equation}
  \binom{n}{ a_{11}n, a_{12}n,\ldots, a_{qq}n} \sim (2\pi n)^{(1-q^2)/2}\, \exp\brk{n H(a)}\, 
  \prod_{i,j\in [q]} (a_{ij}\vee 1/n)^{-1/2}.\label{intermediate}
\end{equation}
Since $1/n\leq a_{ij}\vee 1/n \leq 2/q$ for all $i,j\in [q]$, and since each row and column sum equals
$1/q + o(1)$, the product over $i,j\in [q]$ in (\ref{intermediate}) is always bounded below by a constant
and (easily) bounded above by $O(n^{(q^2-1)/2})$. Therefore
\[ \Omega(n^{(1-q^2)/2})\, \exp\brk{n H(a)} \leq \binom{n}{a_{11}n, a_{12}n,\ldots, a_{qq}n} 
     = O(1)\, \exp\brk{n H(a)}.\]
Combining the above leads to 
\[ \Omega(n^{(1-q^2)/2})\, \sum_{a\in\mathcal{R}} \exp\brk{n F(a)} \leq \Erw[\Zkb^2]
   \leq O(1)\, \sum_{a\in\mathcal{R}} \exp\brk{n F(a)}.
\]
Taking just the term corresponding to $a=a^\ast$ in the lower bound gives
the lower bound of
(\ref{eqLemma_DFGsecond}), and the upper bound follows using the fact that 
$|\cR|\leq n^{\cols^2}$. 

Next, observe that 
if $a\in\mathcal{A}$ then $\prod_{i,j\in [q]} (a_{ij}\vee 1/n)^{-1/2} = \Theta(1)$.
Substituting this into (\ref{intermediate}) and restricting the sum in
(\ref{Esum}) to $\mathcal{A}$
completes the proof of (\ref{eqLemma_DFGsecond_precise}). 
\end{proof}

% For C3
Let $\cD\subseteq \mathbb{R}^{q^2}$ be the polytope comprising of all $a=(a_{ij})_{i,j\in\cols}$ such that
	\begin{align*}
	\sum_{j\in\brk\cols}a_{ij}&=\sum_{j\in\brk\cols}a_{ji}=1/\cols\quad\mbox{for all }i\in\brk\cols,
	&a_{ij}&\geq0\quad\mbox{for all }i,j\in[\cols].
	\end{align*}
Then $\cD$ is the Birkhoff polytope, scaled by a constant factor, and
$\cR\cap\cD$ is dense in $\cD$ as $n\to\infty$.
Therefore, (\ref{eqLemma_DFGsecond}) yields
	$$\lim_{n\to\infty}\frac1n\ln\Erw[\Zkb^2]=\max_{a\in\cD}F(a).$$
Further, evaluating the function $F(a)$ from \Lem~\ref{Lemma_DFGsecond} at the ``flat'' overlap $\bar a=(\bar a_{ij})$ with
$\bar a_{ij}=\cols^{-2}$ for all $i,j\in[\cols]$, we find
\begin{equation} 
\label{Fbara}
 F(\bar a)=2\brk{\ln\cols+c\ln(1-\cols^{1-\unif})}.
\end{equation}
This term { is precisely twice}  
the exponential order of the first moment from (\ref{eqLemma_DFGfirst}).
Consequently, the second moment bound $\Erw[\Zkb^2]=O(\Erw[\Zkb]^2)$ can hold {\em only if}
	\begin{equation}\label{eqsmmMax}
	F(\bar a)=\max_{a\in\cD}F(a).
	\end{equation}
In fact, the Laplace method applied along the lines of~\cite[Theorem 2.3]{GJR} shows that the condition (\ref{eqsmmMax}) is both necessary and sufficient for
the success of the second moment method.
In summary, the second moment argument reduces to the analytic problem of maximising the function $F$ over the polytope $\cD$. 

\paragraph{A relaxation.}
This maximisation problem is anything but straightforward.
Following~\cite{AchNaor}, Dyer, Frieze and Greenhill~\cite{DFG} consider a relaxation.
Namely, instead of optimising $F$ over $\cD$, they consider the (substantially) bigger domain $\cS$ of all
$a=(a_{ij})_{i,j\in[\cols]}$ such that
	$\sum_{j=1}^\cols a_{ij}=1/\cols$ for all $i\in\brk\cols$ and
	$a_{ij}\geq0$ for all $i,j\in[\cols]$,
dropping the constraint that the ``column sums'' $\sum_j a_{ji}$  equal $1/q$. 
Note that $\mathcal{S}$ is the set of singly (row) stochastic matrices, scaled by a constant factor.
%% See C4, for the above
Clearly,
	$\max_{a\in\cD}F(a)\leq\max_{a\in\cS}F(a)$.
Furthermore, Dyer, Frieze and Greenhill solve the latter maximisation problem precisely
by generalising the techniques from~\cite{AchNaor}, 
{requiring rather lengthy technical arguments}.
The result is that for $c$ up to the lower bound in (\ref{eqDFG}) we indeed have $\max_{a\in\cS}F(a)=F(\bar a)$.

But this method does not work up to the density promised by \Thm~\ref{Thm_main}.
	There are two obstacles.
	First, not far beyond the lower bound in~(\ref{eqDFG}) the maximum of $F$ over $\cS$ is attained at a point $ a'\in\cS\setminus\cD$, i.e.,
	 $F( a')>F(\bar a)$.
	Thus, relaxing $\cD$ to the larger domain $\cS$ gives too much away.
	Second, there exists a constant $\gamma>\ln2$ such that for $c=(q^{k-1}-1/2)\ln q-\gamma$, the value of $F$ attained	at 
	\[ a_{\mathrm{stable}}=(\cols^{-1}-\cols^{-\unif})\, \id+\cols^{-\unif}(\cols-1)^{-1} \big(q^2\bar a-\id\big) \in\cD\]
is strictly greater than $F(\bar a)$.
%% For C5
(Note that every entry of the matrix $q^2\bar a$ equals 1.)
	Consequently, even if we could solve the analytic problem of maximising $F$ over the actual domain $\cD$
	 it would be insufficient to prove \Thm~\ref{Thm_main}.

\paragraph{Tame colorings.}
The above discussion shows that it is impossible to prove \Thm~\ref{Thm_main} via the second moment method applied to $\Zkb$.
A similar problem occurs in the case of random graphs ($\unif=2$), see ~\cite{Danny}.
To remedy this problem in the hypergraph case we will generalise the strategy from~\cite{Danny}.

The key idea is to introduce a random variable {$\Zkg$}
that takes the 
% For C6
typical geometry of the set $\COLS$ of all balanced
$\cols$-colorings of {$\hyp\in\hnc$} 
into account, {such that $0\leq\Zkg\leq\Zkb$}. 
According to predictions based on non-rigorous physics considerations~\cite{pnas}, the set
$\COLS$ 
has a geometry that is very different from that of a random subset  of the cube $[\cols]^{\brk{n}}$ of the same size.
More precisely, {for almost all $k$-uniform hypergraphs $\fhyp$ with $cn$
edges}, the set $\COLSH$ 
 decomposes into well-separated ``clusters'' which each 
contains an exponential number of colorings.
However, the fraction of colorings that any single cluster contains is only an exponentially small fraction of the total number of $\cols$-colorings {of $\fhyp$}.
Furthermore, while it is possible to walk inside the set $\COLSH$ from any coloring 
to any other colouring in the same cluster by only changing the {colors}
of $O(\ln n)$ vertices at a time, it is impossible to get from one cluster to another without changing the colors of $\Omega(n)$ vertices in a single step.
Now, the basic idea is to let
$\Zkg=\Zkb\cdot\vecone\{\cT\}$, where $\cT$ is the event that the geometry of the set $\COLSH$ has the aforementioned properties.

To make this rigorous, we define the {\em cluster} of a $\cols$-coloring $\sigma$ of a hypergraph {$\fhyp$} as the set
	\begin{align*}
	\clusterH &=\cbc{\tau\in\COLSH\, :\, \min_{i\in\brk\cols}{a_{ii}(\sigma,\tau)}>\cols^{-1}\myconst}.
	\end{align*}
In words, $\clusterH $ contains all balanced $\cols$-colorings $\tau$ { of $\fhyp$} where, for each color $i$, at least a $\myconst$ fraction of all
vertices colored $i$ under $\sigma$ retain color $i$ under $\tau$.
Call a $\cols$-coloring $\sigma$ of $\fhyp$ {\em separable} if 
	\begin{align}\label{eqsep}
	\forall\tau\in\COLSH,\, \, \forall i,j\in\brk\cols,\,\,  a_{ij}(\sigma,\tau)\not\in(\cols^{-1}\myconst,\cols^{-1}(1-\kappa))\qquad\mbox{where }\kappa=\cols^{1-\unif}\ln^{20}\cols.
	\end{align}

\begin{definition}\label{Def_tame}
A $\cols$-coloring $\sigma$ of the {(fixed)} hypergraph {$\fhyp$} is 
{\em tame} if
	\begin{align*}
		\text{\textbf{T1:}\, $\sigma$ is balanced,\hspace{1cm}
			\textbf{T2:}\, $\sigma$ is separable, \hspace{1cm}
			\textbf{T3:}\, $|\clusterH|\leq
				\Erw[\Zkb].$}
	\end{align*}
\end{definition}
\noindent
\Def~\ref{Def_tame} generalises the concept of ``tame graph colorings'' 
from~\cite[\Def~2.3]{Danny}.

The set of tame colorings of {a given hypergraph $\fhyp$} decomposes into well-separated clusters.
Indeed, the separability condition ensures that the clusters of two tame colorings $\sigma,\tau$ {of $\fhyp$} are either disjoint or identical.
%% For A1
Furthermore, {\bf T3} ensures that no cluster size exceeds the expected number of balanced colorings, i.e., the clusters are ``small''. This will allow us
to control the contribution to the second moment from colourings which lie
in the same cluster (see Lemma~\ref{secmomlaplace_2}).
%% For A2
Furthermore, if $\sigma,\tau$ are tame colorings
then the overlap $a(\sigma,\tau)$ cannot
equal the matrix $a_{\mathrm{stable}}$ defined above, as this matrix
fails {\bf T2}.  So, restricting attention to tame colourings
excludes the matrix $a_{\mathrm{stable}}$.

Let $\Zkg$ be the number of tame $\cols$-colorings {of $\hyp(n,k,cn)$}.
With the right random variable in place, our task boils down to calculating the first and the second moment.
In \Sec~\ref{Sec_first} we will prove that the first moment of $\Zkg$ is asymptotically equal to the first moment of $\Zkb$. {For the following two propositions we assume that }
\begin{align*}
{(\cols^{\unif-1}-1/2)\ln\cols-2\leq c\leq (\cols^{\unif-1}-1/2)\ln\cols-\ln2-{{1.01\ln \cols}/{\cols}}.}
\end{align*}
% For C7, reworded slightly.  OK?
That is, we consider values of $c$ which lie between the standard 
second-moment lower bound (on the $q$-colorability threshold $c_{q,k}$) and the one we prove here.

\begin{proposition}\label{Prop_first}
There is a number $\cols_0>0$ such that for all $\cols>\cols_0$ 
we have
$\Erw[\Zkg]\sim\Erw[\Zkb].$
\end{proposition}

\noindent
Further, in \Sec~\ref{Sec_second} 
we establish the following bound on the second moment.

\begin{proposition}\label{Prop_second}
There is a number $\cols_0>0$ such that for all $\cols>\cols_0$,  
if $\,\Erw[\Zkg]\sim\Erw[\Zkb]$ then 
% For C10
\[ \Erw[\Zkg^2]=O(\Erw[\Zkb]^2) = O(\Erw[\Zkg]^2).\]
\end{proposition}

Thus, while moving to tame colorings has no discernible effect on the first moment, \Prop~\ref{Prop_second} shows that the impact on the second moment is dramatic.
Indeed, {the matrix $a_{\mathrm{stable}}$ shows}  
that $\Erw[\Zkb^2]\geq\exp(\Omega(n))\Erw[\Zkb]^2$ for $c$ near the bound in \Thm~\ref{Thm_main},
while $\Erw[\Zkg^2]=O(\Erw[\Zkb]^2)$ for all $c$ up to {$(\cols^{\unif-1}-1/2)\ln \cols - \ln 2 - {1.01\ln \cols/\cols}$}.
Then \Thm~\ref{Thm_main} follows 
%% Reworded, for C10
from applying Fact~\ref{Fact_smm} to $\Zkg$, 
by \Prop s~\ref{Prop_first} and~\ref{Prop_second}.

Finally, the obvious question is whether the approach taken in this work can
be pushed further to actually obtain tight upper and lower bounds on the $\cols$-colorability threshold.
{However}, it follows from the proof of \Prop s~\ref{Prop_first} that the answer is ``no''. 
More specifically, in \Sec~\ref{Sec_Cor_first} we prove the following.

\begin{corollary}\label{Cor_first}

For any {$\unif\geq3$} there exists a sequence $(\eps_{\cols})_{\cols\geq3}$ such that $\lim_{\cols\to\infty}\eps_{\cols}=0$
and such that the following is true:
For any $c$ such that
\[ (\cols^{\unif-1}- 1/2)\ln\cols-\ln2+\eps_{\cols} < c < (\cols^{\unif-1}-\ 1/2)\ln\cols
\]
there exists $\delta>0$ such that 
	\begin{align}\label{eqCor_first}
	\lim_{n\to\infty}\pr\brk{\Zkc {<} \exp(-\delta n)\, \Erw[\Zkc]}=1.
	\end{align}
\end{corollary}

\noindent
{Now, assume for a contradiction that there is a random variable $0\leq Z\leq\Zkc$ with the following properties.
First, $Z(G)>0$ only if $G$ is $\cols$-colorable.
Second, $\ln\Erw[Z]\sim\ln\Erw[\Zkc]$ (cf.\ \Lem~\ref{Lemma_DFGfirst} and \Prop~\ref{Prop_first}).
Third, $\Erw[Z^2]=O(\Erw[Z]^2)$.
Then the Paley-Zygmund inequality implies that
$$\lim_{\delta\to 0}\lim_{n\to\infty}\pr\brk{\Zkc\geq\exp(-\delta n)\, \Erw[\Zkc]}\geq
		\lim_{\delta\to 0}\lim_{n\to\infty}\pr\brk{Z\geq\exp(-\delta n)\, \Erw[Z]}>0,$$
in contradiction to (\ref{eqCor_first}).}
{\Cor~\ref{Cor_first} is in line with the physics prediction that the actual $\cols$-colorability threshold
	is preceded by another phase transition called {\em condensation}~\cite{pnas},
	beyond which \whp\ $\Zkc\leq\exp(-\Omega(n))\Erw[\Zkc]$.
	In particular, the lower bound of \Thm~\ref{Thm_main} matches this ``condensation threshold'' up to an error term that tends to $0$
	in the limit of large $\cols$.}

%% Notation paragraph used to be here: moved for C8.

\section{The first moment}\label{Sec_first}

\noindent
Throughout this section, 
unless specified otherwise we take $\sigma, \tau:[n]\rightarrow[\cols]$ as balanced maps, and assume that
% I labelled the following, for ease of reference. 
\begin{align}
(\cols^{\unif-1}-1/2)\ln\cols-2\leq c\leq (\cols^{\unif-1}-1/2)\ln\cols-\ln2-{{1.01\ln \cols}/{\cols}}.
\label{c-bounds}
\end{align}
We frequently make use of the Chernoff bound.
\begin{lemma}\label{LemA2}
	(\cite[Theorem 2.1]{JLR}) Let $\phi(x)=(1+x)\ln (1+x) -x$. Let X be a binomial random variable with mean $\mu>0$. Then for any $t>0$
	\begin{align*}
		\textbf{P}\left[X>\mu+t \right]\leq \exp\left\{-\mu\phi(
			t/\mu)\right\},
		\hspace{1cm}
		\textbf{P}\left[X<\mu-t \right]\leq \exp\left\{-\mu\phi(-
			t/\mu)\right\}.
	\end{align*}
In particular, for any $t>1$ we have $\textbf{P}\left[X>t\mu \right]\leq \exp\left\{-t\mu\ln(t/e)\right\}.$
\end{lemma}

\subsection{The planted model}
The aim in this section is to establish \Prop~\ref{Prop_first}, the lower bound on the expected number of tame colorings.
Let $\sigma:\brk n\to\brk\cols$ be a (fixed) balanced map that assigns each vertex a color.
It suffices to prove that $\pr\brk{\sigma\mbox{ is a tame coloring of }\hyp|\sigma\mbox{ is a coloring of }\hyp}=1-o(1)$.
Furthermore, the conditional distribution of $\hyp$ given that $\sigma$ is a coloring admits an easy explicit description:
	the conditional random hypergraph simply consists of $m$ random edges chosen uniformly out of all edges that are not monochromatic under $\sigma$.

It will however be convenient to work with a slightly different distribution.
Let $\phyp\in\plantedhyp$ be the hypergraph on $\brk n$ obtained by including every edge that is not monochromatic under
$\sigma$ with probability $p$, independently, where
% This was incorrect before: there may be some variation in the sizes of the colour classes.
	\begin{align}
		p=
		&\frac{cn}{\binom{n}{\unif}- \prod_{j=1}^q \binom{|\sigma^{-1}(j)|}{k} } \, \sim \, 
		 \frac{c\unif!\cdot\big(1+O(1/n)\big)}{n^{\unif-1}(1-\cols^{1-\unif})}=O\big(n^{1-\unif}\big).
\label{pdef}
	\end{align}
Observe that the expected number of edges equals $cn$.
We call $\plantedhyp$ the {\em planted coloring model}.

\begin{lemma}\label{Lemma_planted}
{Let $\sigma: \brk n \to\brk\cols$ be a fixed balanced map.}
For any event $\cE$ we have
	\begin{align*}
	\pr\brk{\hyp\in\cE \mid \sigma\mbox{ is a coloring of }\hyp}\leq O(\sqrt n)\, \pr\brk{   {\phyp\in\cE}}.
	\end{align*}
\end{lemma}

\begin{proof}
By Stirling's formula, the probability that $\phyp$ has precisely $m$ edges is $\Theta(n^{-1/2})$.
If this event occurs then the {conditional} distributions of $\phyp$ and of $\hyp$ coincide.
\end{proof}

\noindent
Hence, we are left to show that the probability that $\sigma$ fails to be tame in $\phyp$ is $o(n^{-1/2})$.
Indeed, in \Sec s~\ref{Sec_Separability} and~\ref{Sec_clusterSize} we will establish the following two statements.
In both cases the proofs are by careful generalisation of the arguments from~\cite{Danny} to the hypergraph case.

\begin{lemma}\label{Cor_manyNeighbors}
With probability $1-\exp(-\Omega(n))$ the planted coloring $\sigma$ is separable in $\plantedhyp$.
\end{lemma}

\begin{lemma}\label{Cor_clusterSize}
With probability  $1-o(n^{-1/2})$ we have
	$|\cC(\phyp,\sigma)|\leq
				\Erw[\Zkb].$
\end{lemma}

\noindent
\Prop~\ref{Prop_first} is immediate from \Lem s~\ref{Lemma_planted}--\ref{Cor_clusterSize}.

Much of the analysis in this section will involve random variables defined using the following edge counts.
For sets $X_1,X_2,X_3\subset\brk n$ and $\alpha\in[\unif]$, we let $m_{\alpha}(X_1,X_2,X_3)$ be the number of edges $e$ 
of $\phyp$
such that there exists $x\in X_1$ and $v_1,\dots v_\alpha\in X_2$ distinct from one another and from $x$, such that $x,v_1,\dots v_\alpha\in e$ and $e\setminus \{x,v_1,\dots v_\alpha\}\subseteq X_3$. 
If  $\alpha=\unif-1$ then we write $m_{k-1}(X_1,X_2)$ instead of $m_{k-1}(X_1,X_2,X_3)$, since $X_3$ has no effect in this case. For ease of notation, if
$X_1 = \{v\}$ we simply write $m_\alpha(v,X_2,X_3)$, or $m_{k-1}(v,X_2)$. We set $V_i=\sigma^{-1}(i)$ to ease the notational burden.  The following lemmas will be useful later.
Recall that $\kappa = q^{1-k}\ln^{20} q$, as in the definition of separability.

%% CSG 29/10: decided we couldn't have $k$, $\kappa$ and $K$ in the same proof.  Too easy to
%% confuse them.  We don't seem to have used \lambda, so I replaced $K$ by $\lambda$.

\begin{lemma}\label{ABsets}
	For  all sets $A,B\subseteq [n]$ such that $|A|,|B|\leq \frac{n\kappa}{e}$ we have $m_1(A,B,V_i)<20k(|A|+|B|)$ with probability $1-O(1/n)$.
\end{lemma}
\begin{proof}
	Fix sets $A,B$ such that $|A|=a, |B|=b$. Using (\ref{c-bounds}) and (\ref{pdef}), there exists a constant $\lambda \leq \kappa^{-2/3}$  such that
%% CSG 29/10: replaced n/q by |V_i| below.
	\begin{align*}
	\mathbb{P}\left[	m_1(A,B,V_i)>20k(|A|+|B|)\right]&\leq {ab{|V_i|\choose k-2}\choose 20k(a+b)}p^{20k(a+b)}\leq \left(\frac{epab\, |V_i|^{k-2}}{20k(a+b)(k-2)!}\right)^{20k(a+b)}
	%\leq \left(\frac{ec(k-1)q^{2-k}}{20n}\cdot\frac{ab}{a+b}\right)^{20k(a+b)}
	%\leq \left(\frac{e(k-1)q\ln q}{20n}\cdot\frac{ab}{a+b}\right)^{20k(a+b)}
	\leq \left(n^{-1}\lambda \cdot\frac{ab}{a+b}\right)^{20k(a+b)}.
	\end{align*}
Summing over all choices for $A,B$, it follows from the union bound that the probability that any such pair of sets exists is at most
%% CSG 29/10: replaced "given by" by "at most" in the previous sentence
	\begin{align*}
	\sum_{a,b=1}^{n\kappa\pa{/e}} {n\choose a}{n \choose b}\left(n^{-1}\lambda \cdot \frac{ab}{a+b}\right)^{20k(a+b)}&\leq\sum_{a,b=1}^{n\kappa\pa{/e}}\left(\frac{ne}{a}\right)^a\left(\frac{ne}{b}\right)^b\left(n^{-1}\lambda \cdot \frac{ab}{a+b}\right)^{20k(a+b)}.
	%	\\&\leq\sum_{a,b=1}^{n\kappa}\left(\frac{ne}{a}\right)^{a+b}\left(n^{-1}\lambda \cdot \frac{ab}{a+b}\right)^{50k(a+b)}
	%\\&\hspace{-7cm}\leq  \sum_{a,b=1}^{n\kappa}
	%	\left(\frac{ne}{a}\left[n^{-1}\lambda \cdot \frac{ab}{a+b}\right]^k\right)^a\left(\frac{ne}{b}\left[n^{-1}\lambda\cdot \frac{ab}{a+b}\right]^k\right)^b
	%	\left(n^{-1} \lambda\cdot \frac{ab}{a+b}\right)^{49k(a+b)}
	\end{align*}
Now
	\begin{align*}
	\left(\frac{ne}{a}\right)^a\left(\frac{ne}{b}\right)^b\left(n^{-1}\lambda \cdot \frac{ab}{a+b}\right)^{k(a+b)}&=\left[\frac{ne}{a}\left( n^{-1}\lambda \cdot \frac{ab}{a+b}  
	\right)^k \right]^a\left[\frac{ne}{b}\left(n^{-1}\lambda \cdot \frac{ab}{a+b}
	\right)^k\right]^b
	\\&\leq \left[\frac{ne}{a}\left( n^{-1}\lambda \cdot a  
	\right)^k \right]^a\left[\frac{ne}{b}\left(n^{-1}\lambda \cdot b
	\right)^k\right]^b
	%\\&= \left[e\lambda^k\cdot \left( n^{-1} a  
	%\right)^{k-1} \right]^a\left[e\lambda^k\cdot \left(n^{-1} b
	%	\right)^{k-1}\right]^b
	\\&\leq\left[e\lambda^k\cdot \left( \kappa/e
	\right)^{k-1} \right]^a\left[e\lambda^k\cdot \left(\kappa/e
	\right)^{k-1}\right]^b<1.
	\end{align*}
Therefore
	\begin{align*}
	\sum_{a,b=1}^{n\kappa\pa{/e}} {n\choose a}{n \choose b}\left(n^{-1}\lambda\cdot \frac{ab}{a+b}\right)^{20k(a+b)}&\leq	\sum_{a,b=1}^{n\kappa\pa{/e}} \left(n^{-1}\lambda\cdot \frac{ab}{a+b}\right)^{19k(a+b)}
\leq n\kappa\sum_{a=1}^{n\kappa\pa{/e}} \left(n^{-1}\lambda\right)^{19ka}=O(1/n)
	\end{align*}
	where the last equality follows since the summand is decreasing in $a$ when  $a\leq  \frac{n\kappa}{e}$.
\end{proof}
\medskip
\begin{lemma}\label{AzLemma}
\pa{	For every $S\subseteq V$, define}
\begin{align*}\pa{
	B_S:=\{v\in V:m_1(v,S\cap V_j,V_j)>0\text{ for some }j\ne\sigma(v)\}}.
\end{align*}
\pa{	With probability $1-\exp\{-\Omega(n)\}$, every set $S$ of size $|S|\leq nq^{-9k}$ has $|B_S|\leq nq^{-6k}$.}
\end{lemma}
%\begin{lemma}\label{AzLemma}
%	\pa{	With probability $1-\exp\{-\Omega(n)\}$ there is no set $A$ of size $|A|\leq nq^{-9k}$ such that there exists}
%	\begin{align*}\pa{
%		B_A:=\{v\in V:m_1(v,A\cap V_j,V_j)>0\text{ for some }j\ne\sigma(v)\}}
%	\end{align*}
%	\pa{		with $|B_A|\geq  nq^{-6k}$.}
%\end{lemma}
\begin{proof}
	\pa{Fix a subset $S\subseteq V$ of size at most $nq^{-9k}$ and take $v\in V\backslash S$, and some $j\ne\sigma(v)$. Now $m_1(v,S\cap V_j,V_j)$ is stochastically dominated by Bin$\left(|S|{|V_j|\choose k-2},p\right)$. Therefore, the union bound in conjunction with (\ref{c-bounds}), (\ref{pdef}) gives}
	\begin{align*}
	\pa{\textbf{P}\left[m_1(v,S\cap V_j,V_j)>0\text{ for some }j\ne \sigma(v)\right]}&\pa{\leq 1-q\cdot\exp\left\{-p\cdot|S|\cdot{\max_{j\ne\sigma(v)}|V_j|\choose k-2}\right\}
	}\\&\hspace{-1cm}\pa{\leq 1-q	\cdot\exp\left\{-(1+o(1))\cdot\frac{k!\cdot q^{k-1}\ln q }{n^{k-1}(1-q^{k-1})}\cdot \frac{(n/q)^{k-2}}{(k-2)!}
		\cdot\frac{n}{q^{9k}}\right\}
		\leq q^{-8k}.}
	\end{align*}
	\pa{With $B_S$ defined above, it follows that $\vert B_S\vert$ is stochastically dominated by Bin$(n,q^{-8k})$, and so from the Chernoff bound (see Lemma \ref{LemA2}), we have}
	\begin{align*}
	\pa{	\textbf{P}\left[\vert B_S\vert>nq^{-7k}\right]\leq
		\exp\left\{-nq^{-7k}\ln(q^k/e)\right\}\leq \exp\left\{-nq^{-7k}\right\}.
	}
	\end{align*}
	\pa{	Finally, for $\alpha \leq q^{-9k}$  let $X_\alpha$ be the number of sets $S$ of size $\alpha n$ such that $\vert B_S\vert\geq nq^{-7k}$. Then}
	\begin{align*}
	\pa{	\textbf{P}\left[X_\alpha>0\right]\leq 
		\textbf{P}[X_{q^{-9k}}>0]\leq {n\choose q^{-9k} n}\cdot\exp\left\{-nq^{-7k}\right\}
   \leq \exp\left\{-n\left(q^{-9k}\big(q^{2k} - qk\ln q - 1\big)\right)\right\}.}
%\ln q^{-9k})-q^{-7k}\right)\right\}.}
	\end{align*}
	\pa{%For small values of $x$, we have $x(1-\ln x)\leq x^{3/4}$. 
Therefore  for sufficiently large $q$ we  have $\textbf{P}\left[X_\alpha>0\right]\leq \exp\left\{-\Omega(n)\right\}$. The claim follows from taking the union bound over all $\alpha \leq q^{-9k}$ such that $\alpha n$ is an integer: the number of terms in the summation is linear and so is absorbed by the exponential small probability.}
\end{proof}
\subsection{Separability: proof of \Lem~\ref{Cor_manyNeighbors}}\label{Sec_Separability}
Let $\tau:\brk n\to\brk\cols$ be a {balanced} map 
% For C11
which is not separable: that is, for which there exist $i,j\in\brk\cols$ such that~(\ref{eqsep}) is violated.
Of course, we may assume without loss that $i=j=1$.
We aim to show that $\tau$ is unlikely to be a coloring of $\phyp$.
Clearly, {if $\tau$ is a coloring of $\phyp$}  
then $\tau^{-1}(1)$ is an independent set of size about $n/\cols$ that has a rather substantial intersection with the independent set $\sigma^{-1}(1)$.
% For C12, Ok?
Here, as for graphs, an independent set is a set of vertices which contains no edge.
The following lemma rules this constellation out for a wide range of intersection sizes.

\begin{lemma}\label{Lemma_independentSets}
With probability $1-\exp(-\Omega(n))$ the hypergraph $\phyp$ has 
	{no independent set} $I$ {of order $(1{+o(1)})\tfrac{n}{\cols}$} such that
\begin{align*}
	n^{-1}|I\cap\sigma^{-1}(1)|\in\left(\cols^{-1}(1.01/\unif)^{1/(\unif-1)},\cols^{-1}(1-\cols^{(1.01-\unif)/2}) \right).
	\end{align*}
\end{lemma}
\begin{proof}
Suppose that $I$ is an independent set with $|I| = \tfrac{n}{q}(1 + o(1))$ such that $S = I\cap \sigma^{-1}(1)$
contains $|S|= \frac{sn}{\cols}$ vertices, 
% For A3
for some $s\in (0,1)$.
% For B2: I decided it was a mistake to drop the dependence on $S$ from this notation.
Then the set
\begin{align*}
	V_{0}(S):=\{v\in V\backslash \sigma^{-1}(1):m_{\unif-1}(v,S)=0\}
\end{align*}
contains $I\setminus S$.  
Observe that
	\begin{align*}
		\textbf{P}[m_{\unif-1}(& v,S)=0] 
		=\exp\left\{-p\binom{|S|}{\unif-1}\right\}\cdot{\big(1+O(1/n)\big)}
		=\exp\left\{- \frac{\unif c\left(s/\cols \right)^{\unif-1}}{(1-\cols^{1-\unif})}
			\right\}\cdot\big(1+O(1/n)\big)
		\leq 2q^{-ks^{k-1}}.
	\end{align*}
Let $n_0(S):=|V_{0}(S)|$, and observe that in order for $I$ to exist, the inequality
$n_0(S) >(1-s{+o(1)})\tfrac{n}{\cols}$ must hold.
Thus it suffices to prove that 
when $(1.01/k)^{1/(k-1)} < s < 1-q^{(1.01-k)/2}$, with probability $1-\exp(-\Omega(n))$ there is 
no subset $S\subseteq \sigma^{-1}(1)$ of
size $sn/q$ with $n_0(S) > (1-s+o(1))\tfrac{n}{q}$. 

Since $n_{0}(S)$ is stochastically dominated by $ \text{Bin}(|V\backslash\sigma^{-1}(1)|,2q^{-ks^{k-1}})$, we have by the Chernoff bound (see Lemma \ref{LemA2}) that
	\begin{align*}
		\textbf{P}\left[n_{0}(S)>(1-s{+o(1)})\tfrac{n}{\cols}\right]
			\leq \exp\left\{-(1-s+o(1))\tfrac{n}{\cols}\ln \left(\frac{1-s}
				{2q^{1-ks^{k-1}} e}\right) \right\}.
	\end{align*}
The number of choices for a subset $S$ of $\sigma^{-1}(1)$ of size $sn/q$ equals
\[ \binom{n/q(1+o(1))}{sn/q} \leq \left( \frac{e}{1-s}\right)^{(1-s+o(1))n/q} 
      = \exp\left\{(1-s+o(1))\tfrac{n}{q} \big(1-\ln(1-s)\right\},
\]
as established in~\cite[equation (A.5)]{Danny}.
Hence, by the union bound over $S$, the probability that such a subset $S$ exists with the desired
lower bound on $n_0(S)$ is at most
	\begin{align*}
		&\exp\left\{-\big(1-s+o(1)\big)\frac{n}{q}\left(\ln\left(\frac{1-s}{2q^{1-ks^{k-1}}e}\right)-1+\ln(1-s)\right)\right\}
		=\exp\left\{\big(1-s+o(1)\big)\frac{n}{q}\ln\left(\frac{2e^2}{q^{ks^{k-1}-1}(1-s)^2}\right)	\right\}.
	\end{align*}
This probability tends to zero if and only if 
\begin{align}\label{stablecond}
	\frac{\sqrt{2}\, e}{q^{(1-k s^{k-1})/2}} < 1-s.
\end{align}	
By convexity, the exponential function on the l.h.s.\ intersects the
linear function on the r.h.s.\ at most twice, and between these two points of 
intersection the linear function is largest.
For sufficiently large $\cols$, explicit calculation shows that the values $s=(1.01/\unif)^{1/(\unif-1)}$ and $s=1-q^{(1.01-k)/2}$ satisfy (\ref{stablecond}). \pa{Therefore, for fixed $s$ such that $(1.01/\unif)^{1/(\unif-1)}\leq s\leq 1-q^{(1.01-k)/2}$, the probability that the set $S$ exists is bounded by $\exp\left(-\Omega(n) \right)$. Finally, as there are only linearly many such choices for $s$ that make $sn/q$ an integer, this completes the proof.}
\end{proof}

\Lem~\ref{Lemma_independentSets} does not quite cover the entire interval of intersections required by (\ref{eqsep}).
To rule out the remaining {subinterval} $(\cols^{-1}(1-\cols^{(1-\unif)/2}),\cols^{-1}(1-\kappa))$ we use an expansion argument.
The starting point is the observation that most vertices that have color $1$ under $\tau$ but not under $\sigma$ are likely to occur in a good number of edges
in which {\em all} the $\unif-1$ other vertices are colored $1$ under $\sigma$. 
% For C13, OK?
We have not attempted to optimise the constants in this lemma.

\begin{lemma}\label{Lemma_manyNeighbors}
Let $\tau:[n]\rightarrow[\cols]$ be a balanced map such that $a_{11}(\sigma,\tau)\in(\cols^{-1}(1-\cols^{(1.01-\unif)/2}),\cols^{-1}(1-\kappa))$.
With probability $1-\exp(-\Omega(n))$,
the random hypergraph $\phyp\in\plantedhyp$ has the following properties:
% We used to have $m_{k-1}(\{v\},\Sigma^{-1}(1))$ but we can drop the set brackets...
\begin{enumerate}
\item  The set $Y:=\{v\in V\setminus\sigma^{-1}{(1)}\, :\, m_{k-1}(v,\sigma^{-1}(1))<15\}$ has size at most
	$n\kappa/(3\cols)$.
\item The set {$U:=\tau^{-1}(1)\backslash(\sigma^{-1}(1)\cup Y)$}
satisfies
		$m_{1}(U,\sigma^{-1}(1)\setminus\tau^{-1}(1),\sigma^{-1}(1))\leq 5|\sigma^{-1}(1)\setminus\tau^{-1}(1)|$.
\end{enumerate}
\end{lemma}
\begin{proof}
By assumption, $|\sigma^{-1}(1)\cap \tau^{-1}(1)| = \tfrac{sn}{q}$ where
$s\in (1-q^{(1.01-k)/2},\, 1 - \kappa)$.
Fix $v\in V\setminus{V_1}$. We know that 
$$m_{k-1}( v,V_1)\sim \text{Bin}\left(\binom{|V_1|}{\unif-1},p\right).$$ 
Therefore
	\begin{align*}
		\textbf{P}\left[m_{k-1}( v,{V_1})<15\right]
			\leq \sum_{j=0}^{14} \binom{\binom{|V_1|}{\unif-1}}{j}\, p^j(1-p)^{\binom{|V_1|}{\unif-1}-j}
			\leq (1-p)^{\binom{|V_1|}{\unif-1}-14}\sum_{j=0}^{14} \frac{\left(\binom{|V_1|}{\unif-1} p\right)^j}{j!}. 
	\end{align*}
Combining (\ref{pdef}) with the lower bound from (\ref{c-bounds}) shows that
$\binom{|V_1|}{\unif-1}p>\unif\ln \cols$,  which in turn implies that
$$
	\textbf{P}\left[m_{k-1}( v,{V_1})<15\right]\leq 
	{3\left(\unif\ln \cols\right)^{14}\cols^{-\unif}}.$$
As the event $\{m_{k-1}( v,{V_1})< 15\}$ occurs independently for all $v\in V\backslash {V_1}$, the total number $Y$ of such vertices is stochastically dominated by Bin$(n(1-1/\cols),{3\left(\unif\ln \cols\right)^{14}\cols^{-\unif}})$. Therefore $\textbf{E}[Y]\leq n\cdot{3\left(\unif\ln \cols\right)^{14}\cols^{-\unif}}$. 
% We need to know value of $\kappa$ when applying Chernoff
Finally, by the Chernoff bound (see Lemma \ref{LemA2}) and using the definition of $\kappa$ from (\ref{eqsep}),
	\begin{align*}
% Should be strict
		\textbf{P}[Y> n\kappa/(3\cols)]\leq \exp\{\pa{-} n\kappa/(3\cols)\}=\exp\{-\Omega(n)\}
	\end{align*}
and so the proof of $(i)$ is complete. 

For notational convenience, we write $R=\sigma^{-1}(1)\backslash \tau^{-1}(1)$ and $T=\tau^{-1}\pa{(1)}\,\backslash\,\sigma^{-1}(1)$.
Observe that
$m_{1}(U,R,V_1)$  is stochastically dominated by $m_1(T, R, V_1)$, since $U$ is a subset of $T$.
Furthermore, $m_1(T,R,V_1)$ \pa{is stochastically dominated by} Bin$\left(|T| |R|\binom{|V_1|}{\unif-2},p\right)$. Therefore \pa{since $q$ is large with respect to $k$ and $c<q^{k-1}\ln q$, it follows that}
	\begin{align*}
	   \pa{ \binom{|V_1|}{\unif-2}p\leq (1+o(1))\cdot\left(\frac{n}{q}\right)^{k-2}\hspace{-2mm}\cdot \frac{ck(k-1)}{n^{k-1}(1-q^{1-k})}\leq 
	   	 (1+o(1))\cdot n^{-1} q\ln q\cdot \frac{k(k-1)}{(1-q^{1-k})}\leq 
	   	 n^{-1}k^2q\ln q,}
	\end{align*}
\pa{and so}
	\begin{align*}
  \textbf{E}[m_{1}(T,R,V_1)]
  \leq |T| |R| \binom{|V_1|}{\unif-2}p
  & \leq (1+o(1))\, |R|\cdot\,
    \frac{n(1-s)}{q}\,\cdot n^{-1}k^2q\ln q\\
       &= (1+o(1))\,|R|\cdot(1-s) k^2\ln q. 
 \end{align*}
Finally, as $\kappa\leq 1-s\leq q^{(1.01-k)/2}$, part (ii) follows from the Chernoff Bound.
\end{proof}

\begin{proof}[Proof of \Lem~\ref{Cor_manyNeighbors}]
Suppose that $\tau$ is a balanced map such that $a_{11}(\sigma,\tau) > q^{-1}\myconst$.  
By Lemma~\ref{Lemma_independentSets}, we may assume that $a_{11}(\sigma,\tau) > q^{-1}(1 - q^{1.01 - k)/2})$.
With $U,Y$ as in \Lem~\ref{Lemma_manyNeighbors}
we have that 
\begin{align*}
	15|U|\leq m_{1}(U,{R},{V_1})\leq 5|\sigma^{-1}(1)\setminus\tau^{-1}(1)|,
	\end{align*}
and so $|U| \leq \tfrac{1}{3} |\sigma^{-1}(1)\setminus \tau^{-1}(1)|
\pa{\,\sim \tfrac{n}{3q}-\tfrac{1}{3}\vert\sigma^{-1}(1)\cap\tau^{-1}(1)\vert}$.
Since $\tau$ is balanced, we have
\begin{align*}\tfrac{n}{q}\sim |\tau^{-1}(1)| \,\pa{\leq}\, |\sigma^{-1}(1)\cap \tau^{-1}(1)| + |U|+|Y|.\end{align*}
Substituting our bound on $\vert U\vert$ from above, and using 
\Lem~\ref{Lemma_manyNeighbors}, implies that
 $na_{11}(\sigma,\tau)=|\sigma^{-1}(1)\cap\tau^{-1}(1)|> {n(1-\kappa)/q}$, as 
required. ({The failure probability $e^{-\Omega(n)}$ equals the sum of the
failure probabilities from Lemma~\ref{Lemma_independentSets} and Lemma~\ref{Lemma_manyNeighbors}.})
\end{proof}

\subsection{The cluster size: proof of \Lem~\ref{Cor_clusterSize}}\label{Sec_clusterSize}
To upper bound the cluster size we will exhibit a large ``core'' of vertices of $\phyp$ that are difficult to recolor.
More specifically, the core will consist of vertices $v$ such that for every color $i\neq\sigma(v)$ there are several edges $e$ {containing $v$} such that  
$e\setminus\cbc v\subset V_i$ and such that all vertices of $e$ belong to the core.
Therefore, if we attempt to change the color of $v$ to $i\neq\sigma(v)$, then it will be necessary to recolor several other vertices of the core.
In other words, recoloring a single vertex in the core leads to an avalanche that will stop only once at least 
$n\cols^{-1}\myconst$ vertices in some color class have been recolored.
Hence, the outcome is a coloring that does not belong to $\cC(\phyp,\sigma)$.

Formally, {given a fixed balanced map $\sigma$ and fixed hypergraph $\fhyp$, the core $V_{\text{core}}$ of $\fhyp$ is defined as the largest subset $V'\subseteq [n]$ of vertices} such that
		\begin{align*}
		m_{k-1}( v,V_i\cap V') \geq  100\unif
\,\,\,
   \text{ for all }\,\, v\in V'\,\, \text{ and all }\,\, i\ne\sigma(v).
	\end{align*}
The core is well-defined; for if $V',V''$ are sets with the property, then so is $V'\cup V''$.

\begin{lemma}\label{Lemma_core}
With probability $1-o(n^{-1/2})$ the random hypergraph $\phyp$ has the following two properties:
\begin{enumerate}
\item[(i)] The core of $\phyp$ {contains} at least $(1-\cols^{1-\unif}\ln^{500\unif}\cols)n$ {vertices}.
\item[(ii)] If $\tau$ is a {balanced} coloring of $ \phyp$ such that $\tau(v)\neq\sigma(v)$ for some $v$ in the core, then $\tau\not\in\cC(\phyp,\sigma)$.
\end{enumerate}
\end{lemma}

\noindent
We proceed to prove \Lem~\ref{Lemma_core}.
To estimate the 
size of the core we consider the following process:
\begin{enumerate}
	\item[\textbf{CR1}] For $i,j\in [\cols]$ and $i\ne j$, let $W_{ij}=\left\{v\in V_i:
	m_{k-1}( v,V_j)< 300\unif  \right\}$, $W_{i}=\cup_{j:j\ne i} W_{ij}$, $W=\cup_i W_i$. 
	\ \item[\textbf{CR2}] For $i\ne j$ let $U_{ij}=\left\{v\in V_i:
	m_{1}( v,W_j,V_j)>100\unif\right\}$, and $U=\cup_{i\ne j}U_{ij}$.
	\item[\textbf{CR3}] Set $Z^{(0)}=U$ and repeat the following for $\ell\in \mathbb{N}$,
	\begin{itemize}
		\item if there is a $v\in V\backslash Z^{(\ell)}$ such that 
			$m_{1}( v,Z^{(\ell)},V_j)>100\unif$ for some $j\ne\sigma(v)$ then 
			take one such $v$ and let 
			$Z^{(\ell+1)}=Z^{(\ell)}\cup \{v\}$;
		\item otherwise, set $Z^{(\ell+1)}=Z^{(\ell)}$.
	\end{itemize}
\end{enumerate}
Let $Z=\cup_{\ell \geq 0}Z^{(\ell)}$ be the final set resulting from \textbf{CR3}. 

\begin{claim}
\label{core-contain}
The set $V\backslash (W\cup Z)$ is contained within the core.
\end{claim}

\begin{proof} 
For a contradiction, let $v\in V\setminus(W\cup Z)$. 
Since $W_j\subseteq V_j$, any edge counted by $m_{k-1}(v,V_j)$ that does not contribute to 
$m_{1}( v,W_j,V_j)$ must have empty intersection with $W_j$. Since $m_{k-1}( v,V_j)\geq 300\unif$ 
but $m_{1}( v,W_j,V_j)\leq 100\unif$, we must have that $m_{k-1}( v,V_j\backslash W_j)\geq 200\unif$. 
Similarly, since $v\not\in Z$ we have
\[ m_{1}( v,Z\cap V_j,V_j)\leq m_1(v,Z,V_j)\leq100\unif,\]
 and therefore $m_{k-1}( v,V_j\backslash (W_j\cup Z)\pa{)}\geq100\unif$. 
Furthermore, this statement holds for all $j\ne \sigma(v)$ and all $v\in V\setminus (W\cup Z)$. 
It follows that the entire set $V\setminus (W\cup Z)$ may be added to the core, which contradicts
maximality unless $V\backslash (W\cup Z)\subseteq V_{\text{core}}$, as required.
\end{proof}

\noindent
We now bound the size of $W,U$ and $Z$.

\begin{claim}\label{sizeofWij}
	Define the function $Q(\cols,\unif)=\cols^{-\unif-1}\ln^{400\unif}\cols$. With probability at least $1-\exp\{-\Omega(n)\}$ we have $|W_{ij}|\leq n\cdot Q(\cols,\unif)$ for all distinct $i,j\in[q]$.
\end{claim}
\begin{proof}
	Fix $v\in V_i$. Due to the independence of edges in ${\plantedhyp}$ we know that $m_{k-1}( v,V_j)$  is distributed binomially with mean $\binom{|V_j|}{\unif-1}p(1+{o(1)})\geq \unif\ln \cols+
O_\cols(\cols^{-1})$.  It follows from Lemma \ref{LemA2} that $\textbf{P}(v\in W_{ij})\leq
\frac{\cols}{3}\cdot Q(\cols,\unif)$ for $v\in {V_i}$ and sufficiently large $\cols$. Therefore  
	$\textbf{E}[|W_{ij}|] \,\pa{\leq}\, \tfrac{n}{3}\cdot Q(\cols,\unif)$. Finally, since $|W_{ij}|$ is distributed binomially, a straightforward application of the Chernoff bound shows that
$
	\textbf{P}[|W_{ij}|\geq n\cdot Q(\cols,\unif)]\leq \exp\{-n\cdot Q(\cols,\unif)\ln(3/e)\} = \exp\{-\Omega(n)\}$.
\end{proof}

\begin{claim}\label{sizeofu}
	We have $|U|\leq n/\cols^{10\unif}$ with probability at least $1-\exp\{-\Omega(n)\}$.
\end{claim}
\begin{proof}
	Fix $v\in V_i$. The quantity $m_{1}( v,W_j,V_j)$ is stochastically dominated by 
$\text{Bin}\left(|W_j| \binom{|V_j|}{\unif-2},p \right).$ Hence, with $Q(\cols,\unif)$ as previously, we know that 
	\begin{align*}
		&\textbf{E}\left[m_{1}( v,W_j,V_j)\hspace{0.1cm}\Big\vert\hspace{0.1cm} |W_j|\leq n\cdot qQ(\cols,\unif)\right]\leq n\unif p \binom{|V_j|}{\unif-2}\cdot qQ(\cols,\unif)
		=
\widetilde{O}_\cols(\cols^{1-\unif}).
	\end{align*}
Applying the Chernoff bound gives $
		\textbf{P}\left[v\in U_{ij}\hspace{0.1cm}\big\vert\hspace{0.1cm} |W_j|\leq n\cdot qQ(\cols,\unif)\right]
			\leq \widetilde{O}_\cols(\cols^{-19\unif}).$
			Then $\vert U_{ij}\vert$, conditional on the event $\vert W_j\vert \leq n\cdot qQ(q,k)$, is stochastically dominated by a binomial random variable with mean $n\cdot {O}_\cols(\cols^{-\pa{15}\unif})$.
%that $\vert U_{ij}\vert$ has \pa{a conditional binomial distribution for given $W_{j}$} with mean less than $n\cdot \widetilde{O}_\cols(\cols^{-\pa{19}\unif})$. 
The Chernoff bound implies that
	\begin{align*}
		\textbf{P}\left[|U_{ij}|> n\cols^{-10\unif}\hspace{0.1cm}\Big\vert \hspace{0.1cm}|W_j|\leq n\cdot qQ(\cols,\unif)\right] \leq \exp\{-\Omega(n)\}.
	\end{align*}
The result follows by Claim~\ref{sizeofWij}.
\end{proof}

\begin{claim}\label{sizeofZ}
	We have $|Z|\leq n/\cols^{9\unif}$ with probability at least $1-\exp\{-\Omega(n)\}$.
\end{claim}

\begin{proof}
	Claim \ref{sizeofu} tells us that $|U|\leq n/\cols^{10\unif}$  with probability $1-\exp\{-\Omega(n)\}$. 
We will condition on this event.
Suppose that $|Z\backslash U|\geq i^{\ast}=n/\cols^{10\unif}$ and
consider the set $Z^{(i^\ast)}$ obtained after $i^\ast$ steps of {\bf CR3}.
The construction of $Z$ implies that there exists $100\unif|Z^{(i^*)}\setminus U|$ vertex-edge pairs $(v,e)$ such that $e\cap Z\geq 2$ and $e\setminus \{v\}\subseteq V_j$ for some $j\in[\cols]$. Since each edge may appear in at most
$\unif$ vertex-edge pairs, this implies that there are at least $100\, |Z^{(i^*)}\setminus U|$ such edges. 
Therefore, there are at least $100i^\ast=100 n/\cols^{10\unif}$ edges $e$ such 
that $e\cap Z^{(i^\ast)}\geq 2$ and $e\backslash\{v\}\subseteq V_j$ for some 
$j\in[\cols],v\in e$, despite the set $Z^{(i^\ast)}$ only being of size 
at most $2n/\cols^{10\unif}$. We prove that with high probability,
no such set can exist.

Let $\alpha=\cols^{-10\unif}$ and let $T\subset [n]$ be a set of $|T|=\alpha n$
vertices. Let $m_{T}^\ast$ be the number of edges $e$ such that  $e\cap T\geq 2$ and $e\backslash\{v\}\subseteq V_j$ for some $j\in[\cols], v\in V$. We know that $m_{T}^\ast$ is stochastically dominated by 
$\text{Bin}\Big(2 \alpha \binom{n}{2} \binom{n/\cols}{\unif-2},p\Big)$, and so we may observe  
by the Chernoff bound that
\begin{align*}
\textbf{P}\left[ m_{T}^\ast\geq 100\cdot \alpha n\right]\leq 
  \exp\{100\alpha n\ln \alpha\}.
	\end{align*}
If we let $N$ be the number of sets $T$ of size $|T|=\alpha n$ such that $m_{T}^\ast\geq  100\cdot \alpha n$, then
	\begin{align*}
	\textbf{P}\left[ N>0\right]\leq \binom{n}{\alpha n}\exp\{100\alpha n\ln \alpha\}\leq \exp\left\{-\Big(\alpha\ln\alpha+(1-\alpha)\ln(1-\alpha)-100\alpha\ln\alpha\Big)n\, \right\}=\exp\{-\Omega(n)\}.
	\end{align*} 
% For A4
The final bound holds since $\alpha$ is constant with respect to $n$
and $\alpha\in (0,1)$.
Therefore with probability $1-\exp\{-\Omega(n)\}$ we have $|Z\backslash U|\leq n/\cols^{10\unif}$, which implies the claim.
\end{proof}

\noindent
Lemma \ref{Lemma_core} $(i)$ then follows from Claims \ref{core-contain}--\ref{sizeofZ}. 

To establish (ii) we say that a vertex $v$ is {\em $j$-blocked} if there is an edge $e\ni v$ such that 
$e\setminus\cbc v$ is contained in the core and $e\setminus\cbc v\subset {V_j}$. We say that  a vertex $v$ is {\em $\sigma$-complete} {if it is}
$j$-blocked {for all $j\neq\sigma(v)$}.
Note that, as with vertices inside the core,
recoloring any $\sigma$-complete vertex will set off a coloring avalanche.

\begin{claim}\label{sigmacompletewhp}
	With probability ${1-O(1/n)}$ the random graph $\phyp$ has the following property:
	\begin{align*}
	\textit{if $\tau\in\mathcal{C}(\phyp,\sigma)$ then for all $\sigma$-complete vertices $v$ we have $\sigma(v)=\tau(v)$.}
	\end{align*}
\end{claim}
\begin{proof}
	Note that it suffices to prove that $\sigma(v) = \tau(v)$ for all $v$ in the
	core, since this implies the result for all $\sigma$-complete vertices
	outside the core as well, by definition of $\sigma$-complete.
	
	Recalling Lemma \ref{Cor_manyNeighbors}, we may assume that $\sigma$ is separable in $\plantedhyp$. For $i\in [q]$, let
	\[
	\Delta_i^+=\{v\in V_{\text{core}}:\tau(v)=i\ne \sigma(v)\},
	\hspace{1cm}
	\Delta_i^-=\{v\in V_{\text{core}}:\tau(v)\ne i=\sigma(v)\}.
	\]
	Then
	\begin{equation}\label{sumofDelta}
	\sum_{i=1}^\cols|\Delta^+_i|=|\{v\in V_{\text{core}}:\sigma(v)\ne \tau(v)\}|=\sum_{i=1}^\cols|\Delta_i^-|.
	\end{equation}

	Since $\sigma$ is separable and $\tau\in\mathcal{C}(\mathcal{G}_\sigma,\sigma)$ we have $\max_{i\in [\cols]}|\Delta_i^+|\leq \tfrac{n}{\cols}\kappa(1+{o(1)})$ and $\max_{i\in [\cols]}|\Delta_i^-|\leq \tfrac{n}{\cols}\kappa(1+{o(1)})$. If we can show that $\{v\in V_{\text{core}}:\sigma(v)\ne \tau(v)\}=\emptyset$ then $\sigma(v)=\tau(v)$ for all $\sigma$-complete vertices.
	
	Take $v\in \Delta_i^+$. Since $v\in V_{\text{core}}$ we know that $m_{\unif-1}( v, V_i){\geq 100k}$. Further, since $\tau$ is a coloring, we must have that $m_{1}( v,\Delta_i^-,V_i) {\geq 100k}$. 
By Lemma~\ref{ABsets} we know that with probability $1-O(1/n)$, 
\begin{equation}
\label{ifthisholds} 
m_{1}(\Delta_i^+,\Delta_i^-,V_i)\leq 20k\left(\vert\Delta_i^+\vert+\vert\Delta_i^-\vert\right) \qquad \text{ for all }
\,\, i\in [q]. 
\end{equation}
Observe that if (\ref{ifthisholds}) holds then for all $i\in [q]$,
\[ 100k|\Delta_i^+|\leq  m_{1}(\Delta_i^+,\Delta_i^-,V_i) \leq 
   20k(|\Delta_i^-|+|\Delta_i^+|).
\] 
But this implies that $4|\Delta_i^+|\leq |\Delta_i^-|$ for all $i\in [q]$, which contradicts
(\ref{sumofDelta}) unless $\Delta_i^+=\Delta_i^-=\emptyset$ for all $i\in [q]$. 
Therefore we conclude that with probability $1-O(1/n)$, for all $v$ in the core we have $\sigma(v) = \tau(v)$,
completing the proof.
\end{proof}
\noindent
\Lem~\ref{Lemma_core} (ii) is immediate from Claim \ref{sigmacompletewhp}.
The core size guaranteed by \Lem~\ref{Lemma_core} is not quite big enough to deduce a good bound on the cluster size
	(due to the {polylogarithmic} factor). 
Recall that a vertex $v$ is {\em $j$-blocked} if it is contained in an edge $e$ such that 
$e\setminus\cbc v$ is contained in the core and $e\setminus\cbc v\subset {V_j}$.
Further, we say that $v$ is {\em $\alpha$-free} if there are at least $\alpha+1$ colors $j$ (including $\sigma(v)$) such that $v$ fails to be $j$-blocked. 
A careful study of how the vertices outside the core connect to those inside yields the following.

\begin{lemma}\label{Lemma_coreConnect}
{With probability $1-\exp\{-\Omega(n)\}$
there exists a set $A_W$ of vertices such that  there are at most $n\cols^{1-\unif}(1 + O_\cols(\cols^{-2}))$ vertices outside of $A_W$ which are $1$-free and there are at most $n\cols^{-\unif}(1 + O_\cols(\cols^{-1}))$ vertices which belong to $A_W$ or are 2-free.}
\end{lemma}

\noindent
We proceed to prove \Lem~\ref{Lemma_coreConnect}.
% We used to have $m_{k-1}(\{v\},\Sigma^{-1}(1))$ but we can drop the set brackets...
Let $$A_{i}=\{v\in V\backslash V_i : m_{k-1}(v,V_i)=0\},$$ and define
	\begin{align*}
		&A_0=\bigcup_{i\in [q]}A_i,
			\hspace{0.75cm}
		A_{00}=\bigcup_{i\ne j}\left(A_i\cap A_j\right),
		\hspace{0.75cm}A_Z=\{v\in V:m_{1}( v,Z\cap V_i,V_i)>0
			\text{ for some }i\ne\sigma(v)\},\\
		&\hspace{3.2cm}A_W = \{ v\in V : m_{k-1}(v,V_i\setminus W_i) = 0
                   \text{ for some }i \neq \sigma(v) \} \setminus A_0\, .
	\end{align*}
 We claim that if $v$ is $1$-free then $v\in A_0\cup A_Z\cup A_W$, and if $v$ is $2$-free then $v\in A_{00}\cup A_Z\cup A_W$. \pa{ To see that this is the case, note that if $v$ is $1$-free then there is some $i\ne \sigma(v)$ such that there is no edge $e\ni v$ with $e\backslash \{v\}\subseteq V_i\cap V_{\text{core}}.$ For a contradiction, suppose that $v\notin A_0\cup A_W \cup A_Z$. Then there must be an edge $e'\ni v$ such that
	\begin{align*}
		  e'\backslash\{v\}\subseteq V_i\backslash (W_i\cup Z)\hspace{0.5cm}
		 \text{ for some }\hspace{0.5cm}i\ne \sigma(v).
	\end{align*}
	However, we know from Claim \ref{core-contain} that $V_i\backslash (W_i\cup Z)\subseteq V_i\cap  V_\text{core}$, giving the desired contradiction. The case for $2$-free vertices is similar. Suppose that $v$ is $2$-free and $v\notin A_{00}\cup A_Z\cup A_W$. Since $v$ is also $1$-free, must have $v\in \left(A_0\backslash A_{00} \right)\cap (A_Z\cup A_W)^c$. Since $v\in A_0\backslash A_{00}$, there exists $i\ne \sigma(v)$ so that $m_{k-1}(v,V_j)>0$ for all $j\notin \{\sigma(v),i\}$. That is to say, we have $v\in (A_j\cup A_Z\cup A_W)^c$ for all $j\notin \{\sigma(v),i\}$. This means that for all $j\notin \{\sigma(v),i\}$ there exists  an edge $e'\ni v$ such that $ e'\backslash\{v\}\subseteq V_j\backslash (W_j\cup Z)$. As above, we conclude that 
	$e'\subseteq V_j\cap  V_\text{core}$ and so $v$ is $j$-blocked for all $j\notin \{\sigma(v),i\}$.	
%	Then, since $v\notin A_Z\cup A_W$, we can run the $1$-free argument for each colour $j\notin \{\sigma(v),i\}$. This shows that $v$ is not $j$-blocked for any $j\ne \{\sigma(v),i\}$. 
Therefore $v$ is not $2$-free and so we have a contradiction.
}

Thus,  to prove Lemma \ref{Lemma_coreConnect} it suffices to bound the size of the sets $A_0, A_{00}, A_W, A_Z$.
%By construction, if $v$ is $1$-free then $v\in A_0\cup A_Z\cup A_W$, and if $v$ is $2$-free then $v\in A_{00}\cup A_Z\cup A_W$. Thus, if we can bound the size of these sets that we will obtain the estimates required by Lemma \ref{Lemma_coreConnect}.
\begin{claim}\label{A0size}
	We have $|A_0|\leq n/\cols^{\unif-1}$ and $| A_{00}|\leq n/\cols^{2\unif-2}$ with probability at least $1-\exp\{-\Omega(n)\}$.
\end{claim}
\begin{proof}
	Take $v\in V_j$, and $i \ne j$. Now $
		\textbf{P}\left[m_{k-1}( v,V_i)=0\right]
			< \exp\left\{-\unif\ln \cols\right\}$,
and hence
$\textbf{P}\left[v\in A_0\right]\leq (\cols-1)\cols^{-\unif}$. It follows that $\textbf{E}[|A_0|]:=\mu<n\cdot  (\cols-1)\cols^{-\unif}$. Since $\textbf{P}\left[m_{k-1}( v,V_i)=0\right]>\exp\{-(\unif+1)\ln\cols\}$ we must have that $\mu>nq^{-\unif}$ 
and so by the Chernoff bound 
	\begin{align*}
		\textbf{P}\left[|A_0|>n/\cols^{\unif-1}\right] 
		\leq \exp\left\{-nq^{-\unif}\left[\frac{\cols}{\cols-1}\ln\left(\frac{\cols}{\cols-1}\right)-\frac{1}{\cols-1} \right]\right\}
		=\exp\{-\Omega(n)\}
	\end{align*}
	as desired. Further, the argument for $A_{00}$ follows quickly after noting that the edge sets of $m_{k-1}(v,V_i)$ and $m_{k-1}( v,V_j)$ are independent for $i\ne j$.
\end{proof}

\begin{claim}\label{Azsize}
	We have $|A_Z|\leq n/\cols^{6\unif}$ with probability at least $1-\exp\{-\Omega(n)\}$.
\end{claim}
\begin{proof}	
	\pa{The proof follows immediately from Lemma \ref{AzLemma} and Claim \ref{sizeofZ}.}
\end{proof}

\begin{claim}\label{Awsize}
	We have $|A_W|\leq n/\cols^{
		\unif}$ with probability at least $1-\exp\{-\Omega(n)\}$.
\end{claim}
\begin{proof}
	{Fix $i\ne j$ and $v\in V_i$. %In what follows we condition the event $\mathcal{E}$ that $W_{j}=nq\cdot Q(q,k)$ where $Q(q,k)=q^{-k-1}\ln^{400k}q$ as in Claim \ref{sizeofWij}. 
	We seek to compute the following probability:
}
\begin{align*}
{		\textbf{P}\left[m_{k-1}(v,V_j\backslash W_j)=0 \,\text{ and }\,m_{k-1}(v,V_j)>0 \right]=\textbf{P}\left[m_{k-1}(v,V_j\backslash W_j)=0 \right]\cdot\textbf{P}\left[m_{k-1}(v,W_j)>0 \right].
}\end{align*}
{	Since $V_j\backslash W_j\subseteq V_j$, we know from the calculations in Claim \ref{A0size} that $\textbf{P}\left[m_{k-1}(v,V_j\backslash W_j)=0 \right]\leq q^{-k}$. Further, define the event $\mathcal{E}$ that $\vert W_{j}\vert\leq nq\cdot Q(q,k)$ where $Q(q,k)=q^{-k-1}\ln^{400k}q$ as in Claim \ref{sizeofWij}. We know that 
$m_{k-1}(v,W_j)	\,\big\vert\,\mathcal{E}$ is stochastically dominated by $\text{Bin}\left({nq\cdot Q(q,k) \choose k-1},p\right)$ and so for sufficiently large $q$, we have
}\begin{align*}
{\textbf{P}\left[m_{k-1}(v,W_j)>0 \,\big\vert\,\mathcal{E}\right]
%=1-(1-p)^{{nq\cdot Q(q,k) \choose k-1}}
\leq p\,{nq\cdot Q(q,k) \choose k-1}
\leq \frac{p\left(nqQ(q,k)\right)^{k-1}}{(k-1)!}\leq
2ck\left(qQ(q,k)\right)^{k-1}
%\frac{ck!}{n^{k-1}}\cdot \frac{n^{k-1}q^{k-1}q^{-(k^2-1)}\ln^{400k^2}q}{(k-1)!}
\leq q^{-k(k-2)}.
}\end{align*}
{Since Claim \ref{sizeofWij} implies that $\mathcal{E}$ only fails with exponentially small probability, it follows that 
}
\begin{align*}
{\textbf{P}\left[m_{k-1}(v,V_j\backslash W_j)=0 \,\text{ and }\,m_{k-1}(v,V_j)>0 \right]
\leq q^{-k}\left( q^{-k(k-2)}\cdot\textbf{P}\left[\mathcal{E}\right]+\textbf{P}\left[\neg\mathcal{E}\right]\right)=q^{-k^2+k}+\exp\left\{-\Omega(n)\right\}.
}
\end{align*}
{
	Taking the union bound over $j\in [q]\backslash \{i\}$ shows that for $v\in V_i$. 
	\begin{align*}
		\textbf{P}[v\in A_W]\leq q^{-k^2+k+1}+\exp\left\{-\Omega(n)\right\}<q^{-k-1}.
	\end{align*}
	Therefore $\vert A_W\vert$ is stochastically dominated by Bin$(n,q^{-k-1})$, and applying Lemma \ref{LemA2} completes the proof.}
\end{proof}

\noindent
Thus Lemma \ref{Lemma_coreConnect} follows from Claims \ref{A0size}-\ref{Awsize}.

\begin{proof}[Proof of \Lem~\ref{Cor_clusterSize}]
{Assume that the properties described in 
Claim~\ref{sigmacompletewhp} and Lemma~\ref{Lemma_coreConnect} both hold, noting that 
this is an event with probability $1-o(n^{-1/2})$.
The remainder of the proof is deterministic.}

{Since we have assumed that Claim~\ref{sigmacompletewhp} succeeds, }
for all $\sigma$-complete $v$ and all $\tau\in\cC(\phyp,\sigma)$ we have $\tau(v)=\sigma(v)$. {Let $F_x$ be the set of $x$-free vertices.  
Next, by our assumption that \Lem~\ref{Lemma_coreConnect} succeeds, we have}
	\begin{align*}
		|F_1\backslash A_W|\leq \frac{n}{\cols^{\unif-1}}+n\cdot O_\cols\left(\cols^{-\unif-1}\right),
		\hspace{1cm}|F_2\cup A_W| = \frac{n}{q^k}+n\cdot O_\cols\left(\cols^{-\unif-1}\right).
	\end{align*}
%	\begin{align*}
%	{|F_1\backslash A_W|\leq \frac{n}{\cols^{\unif-1}}+n\cdot O_\cols\left(\cols^{-\unif-1}\right),
%		\hspace{1cm}|F_2\backslash A_W| = n\cdot O_\cols\left(\cols^{-\unif-1}\right),\hspace{1cm}|A_W|\leq n\cdot \cols^{-\unif}.}
%	\end{align*}
For any $v\in F_x\backslash F_{x+1}$ there are at most $x+1$ choices for the color of $v$. Since $F_{x+1} \subseteq F_{x}$ it follows that
	\begin{align*}
		{|\cC(\phyp,\sigma)|\leq 2^{|F_1\setminus F_2|}
			3^{|F_2\setminus F_3|}
				\cdots
			\cols^{|F_\cols|}
		\leq 2^{|F_1\setminus A_W|}\cdot \cols^{|F_2\cup A_W|},}
	\end{align*}
and so
	\begin{align*}
		{\frac{1}{n}\ln |\cC(\phyp,\sigma)|\leq \frac{\ln 2}{\cols^{\unif-1}}+\frac{\ln \cols}{\cols^{\unif}}+\widetilde{O}_\cols(\cols^{-\unif-1}).}
	\end{align*}
{Furthermore, since $c\leq (\cols^{\unif-1}-1/2)\ln \cols - \ln 2 
- \frac{1.01\ln \cols}{\cols}$, we have }
	\begin{align*}
		{\frac{1}{n}\ln \textbf{E}[Z_{\cols,\text{bal}}] \geq} \ln \cols +c\ln(1-\cols^{1-\unif})=\frac{\ln 2}{\cols^{\unif-1}}+\frac{1.01 \ln \cols}{\cols^\unif}+\widetilde{O}_\cols(\cols^{-\unif-1}).
	\end{align*}
{These bounds imply that $|\cC(\phyp,\sigma)|\leq \textbf{E}[Z_{\cols,\text{bal}}]$, 
completing the proof.}
%Hence by Lemma \ref{Lemma_planted}, \textbf{T3} holds in $\hyp(n,\unif,cn)$ 
%with probability $1-o(1)$, completing the proof of  \Lem~\ref{Cor_clusterSize}.
\end{proof}

\subsection{Proof of \Cor~\ref{Cor_first}}\label{Sec_Cor_first}

\noindent
Here we assume that 
\[ 
 (\cols^{\unif-1}-1/2)\ln \cols-\ln2+1/\ln\cols < c < (\cols^{\unif-1}-1/2)\ln \cols.
\]
The proof of \Cor~\ref{Cor_first} is similar to the proof of \cite[\Prop~2.1]{Danny}.
The starting point is the following observation, which is reminiscent of the ``planting trick'' from~\cite{Barriers}.
Call $\sigma:[n]\to[\cols]$ {\em $\eps$-balanced} for some $\eps>0$ if $\max_{i\in[\cols]}||\sigma^{-1}(i)|-n/\cols|<\eps n$.

\begin{claim}\label{Claim_antiPlant}
Suppose there exist $\eps,\eps'>0$ and a sequence $(\cE_n)_n$ of events such that for large $n$ and all $\eps$-balanced $\sigma:[n]\to[\cols]$
we have
% C15: I just deleted "while"
	\begin{align}
	\pr\brk{\phyp\in\cE_n}& \leq\exp(-\eps'n),\label{eqClaim_antiPlant1}\\
	\lim_{n\to\infty}\pr\brk{\hyp\in\cE_n}&=1.\label{eqClaim_antiPlant2}
	\end{align}
Then there exists $\delta>0$ such that \whp $ $
 	$\Zkc(\hyp)\leq\exp(-\delta n)\Erw[\Zkc(\hyp)]$. 
\end{claim}
\begin{proof}
Let $\Zkc'(G)$ be the number of $\eps$-balanced $\cols$-colorings of $G$.
By \cite[proof of~\Lem~2.1]{DFG} there exists $\alpha>0$ such that
	\begin{equation}\label{eqClaim_antiPlant3}
	\Erw[\Zkc(\hyp)-\Zkc'(\hyp)]\leq\exp(-\alpha n)\Erw[\Zkc(\hyp)].
	\end{equation}
Further, let $\Zkc''(\hyp)=\Zkc'(\hyp)\vecone\{\hyp\in\cE_n\}$.
% For B6
Combining Lemma~\ref{Lemma_planted} and (\ref{eqClaim_antiPlant1}) shows that
	\begin{equation}\label{eqClaim_antiPlant4}
	\Erw[\Zkc''(\hyp)]\leq\exp(-\eps' n/2)\Erw[\Zkc'(\hyp)].
	\end{equation}
Moreover, let $\cA_n=\{\Zkc(\hyp)\geq\exp(-\delta n)\Erw[\Zkc(\hyp)]\}$ for a small enough $\delta>0$.
Combining (\ref{eqClaim_antiPlant3}) and (\ref{eqClaim_antiPlant4}), we obtain
	\begin{align*}
	\exp(-\delta n)\Erw[\Zkc(\hyp)]\pr\brk{\hyp\in\cA_n\cap\cE_n}&\leq\Erw[\Zkc(\hyp)\vecone\{\hyp\in\cA_n\cap\cE_n\}]\\
		&\leq\Erw[\Zkc''(\hyp)]+\Erw[\Zkc(\hyp)-\Zkc'(\hyp)]
			\leq(\exp(-\eps'n/2)+\exp(-\alpha n))\Erw[\Zkc(\hyp)].
	\end{align*}
Hence, choosing $\delta>0$ small enough and recalling (\ref{eqClaim_antiPlant2}), we obtain $\pr\brk{\cA_n}=o(1)$.
\end{proof}

Thus, we are left to exhibit a sequence of events as in Claim~\ref{Claim_antiPlant}.
Given a map $\tau:[n]\to[\cols]$ and a hypergraph $G$ on $[n]$ let $E_{\tau}(G)$ be the number of monochromatic edges of $G$ under $\tau$.
Further, for $\beta>0$ let
	\begin{align*}
	Z_{\cols,\beta}(G)&=\sum_{\tau}\exp(-\beta E_\tau(G)),
	\end{align*}
where the sum ranges over all $\tau:[n]\to[q]$.
The function $\beta\mapsto Z_{\cols,\beta}(G)$ can be viewed as the partition function of a hypergraph variant of the ``Potts antiferromagnet'' from statistical physics.
We consider this random variable because it is concentrated in the following sense.

\begin{claim}\label{Claim_plantedAzuma}
For any $\eps>0$ there is $\delta>0$ such that
for any $\sigma:[n]\to[\cols]$ we have
 \begin{align*}
 \pr\brk{\abs{\ln Z_{\cols,\beta}(\hyp)-\Erw \ln Z_{\cols,\beta}(\hyp)}>\eps n}&<\exp(-\delta n),&
 {\pr\brk{\abs{\ln Z_{\cols,\beta}(\phyp)-\Erw \ln Z_{\cols,\beta}(\phyp)}>\eps n}<\exp(-\delta n).}
 \end{align*}
\end{claim}
\begin{proof}
Either adding or removing a single edge alters the value of $\ln Z_{\cols,\beta}$ by at most $\beta$.
Therefore, the assertion follows from a standard application of Azuma's inequality.
\end{proof}

\noindent
Additionally, we have the following estimate of $\Erw \ln Z_{\cols,\beta}(\phyp)$.

\begin{claim}\label{Claim_quenchedAvg}
There is $\delta>0$ such that for all $\beta>0$ 
and all $\delta$-balanced $\sigma$ 
we have $\Erw \ln Z_{\cols,\beta}(\phyp)>\delta n+\ln\Erw[\Zkc(\hyp)]$.
\end{claim}
\begin{proof}
We are going to show that for a small enough $\delta>0$ we have \whp\
	\begin{align}\label{eqClaim_quenchedAvg1}
	n^{-1}\ln Z_{\cols}(\phyp)\geq q^{1-k}\ln2+\tilde O_q(q^{-k}).
	\end{align}
Since  $Z_{\cols,\beta}(\phyp)\geq Z_{\cols}(\phyp)$ for all $\beta$ and because
(\ref{eqLemma_DFGfirst}) implies that
\[ n^{-1}\ln\Erw[\Zkc(\hyp)]\leq q^{1-k}\big(\ln 2- (2\ln q)^{-1}\big)+\tilde O_q(q^{-k}),
\]  
the claim follows from (\ref{eqClaim_quenchedAvg1}).

To prove (\ref{eqClaim_quenchedAvg1}),
we let $F_{ij}$ be the set of vertices $v\in V_i$ such that {$m_{k-1}(v,V_j)=0$}.
Further, let $F_{ij}'$ be the set of all $v\in F_{ij}$ such that {$m_{k-1}(v,V_h)=0$}
for some $h\in[\cols]\setminus\{i,j\}$.
Due to the independence of the edges, $|F_{ij}|,|F_{ij}'|$ are binomial random variables.
The expected sizes of these sets satisfy
% For B7 and B8
	\begin{align*}
	\Erw|F_{ij}|&
	= \big( q^{-k{-1}}+\widetilde O_\cols(q^{-k{-2}})\big)n,
          &
	\Erw[F_{ij}'|&=
		\widetilde O_\cols(q^{-k{-2}})n.
	\end{align*}
Hence, the Chernoff bound implies that with probability $1-\exp(-\Omega(n))$,  for all $i,j$ we have
	\begin{align}\label{eqClaim_quenchedAvg2}
	|F_{ij}|&=(q^{-k{-1}}+\widetilde O_\cols(q^{-k{-2}}))n,&
	|F_{ij}'|&=		\widetilde O_\cols(q^{-k{-2}})n.
	\end{align}
Let $F_{\star}=\bigcup_{i\neq j}F_{ij}\setminus F_{ij}'$.
Further, for every vertex $v\in F_{\star}$ let $\sigma_{\star}(v)\in[\cols]$ be the (unique) color such that $v\in F_{\sigma(v)\sigma_{\star}(v)}$.
Further, let $E_{\star}$ be the set of edges $e$ of $\phyp$ such that there exist $v,w\in e\cap F_{\star}$ such that
	\begin{align*}
	\sigma(e\setminus\{v,w\})\subset\{\sigma(v),\sigma(w),\sigma_{\star}(v),\sigma_{\star}(w)\}.
	\end{align*}
% For B9
%Given that (\ref{eqClaim_quenchedAvg2}) occurs, 
The random variable $|E_{\star}|$ is stochastically dominated by a binomial random variable
$\operatorname{Bin}(cn, p_0)$ where
\[ p_0=  \frac{2|F_{\star}|^2}{n^2}\, (4 q^{-1})^{k-2}.
\]
So by (\ref{eqClaim_quenchedAvg2}),
\[ \Erw |E_{\star}| =  2( q^{1-k} + \widetilde{O}_q(q^{-k}))^2\, (4q^{-1})^{k-2}\, cn
+ \exp(-\Omega(n)) = \widetilde{O}_q(q^{1-k})\, n.
\]
Then, the Chernoff bound, we find that with probability
$1-\exp(-\Omega(n))$,
	\begin{align}\label{eqClaim_quenchedAvg3}
	|E_{\star}|= \widetilde O_{\cols}(q^{-k})\, n.
	\end{align}
Now, let $F_0$ be the set of all vertices $v\in F_{\star}$ that do not occur in any $e\in E_{\star}$.
Then by construction any map $\tau:[n]\to[\cols]$ such that
	$\tau(v)\in\{\sigma(v),\sigma_\star(v)\}$ for all $v\in F_0$ and
	$\tau(v)=\sigma(v)$ for all $v\not\in F_0$ is a $\cols$-coloring of $\phyp$.
Furthermore, there are
	$2^{|F_0|}$ such $\tau$ and
	 (\ref{eqClaim_quenchedAvg2}), (\ref{eqClaim_quenchedAvg3}) entail that $|F_0|\geq(q^{1-k}+\widetilde O_\cols(q^{-k}))n$ \whp,
	 whence (\ref{eqClaim_quenchedAvg1}) follows.
\end{proof}

\noindent
By comparison,  $\ln \Erw [Z_{\cols,\beta}(\hyp)]$ is upper-bounded as follows.

\begin{claim}\label{Claim_annealedAvg}
For any $\delta>0$ there is $\beta_0>0$ such that for all $\beta>\beta_0$ we have
	$\ln\Erw[Z_{\cols,\beta}(\hyp)]\leq\delta n+\ln\Erw[\Zkc(\hyp)]$.
\end{claim}
\begin{proof}
Using (\ref{eqLemma_DFGfirst}) and the fact that monochromatic edges are least likely when $\tau$ is balanced, we obtain
	\begin{align*}
	\frac{1}{n}\ln \Erw[\Zkc(\hyp)]&=\ln q +c\ln \bc{1-\cols^{1-\unif}}+o(1),&
	\frac{1}{n}\ln \Erw[Z_{\cols,\beta}(\hyp)]&\leq\ln q+c\ln \bc{1-\cols^{1-\unif}(1-\exp(-\beta))}.
	\end{align*}
Making $\beta$ sufficiently large and taking logarithms, we obtain the assertion.
\end{proof}

\noindent
Finally, we know from Claims~\ref{Claim_quenchedAvg}--\ref{Claim_annealedAvg} and Jensen's inequality that there exists $\delta>0$ such that $
	\Erw \ln Z_{\cols,\beta}(\hyp)+\delta n\leq \Erw \ln Z_{\cols,\beta}(\phyp)$.
However, Claim~\ref{Claim_plantedAzuma} implies that both $\ln Z_{\cols,\beta}(\hyp)$ and $\ln Z_{\cols,\beta}(\phyp)$ are close to their expectations. 
% For B10
Therefore  \Cor~\ref{Cor_first} follows by applying Claim~\ref{Claim_antiPlant} to the event
\[  \mathcal{E}_n=\left\{G:\abs{\ln Z_{\cols,\beta}(G)-\Erw \ln Z_{\cols,\beta}(\hyp)}>\eps n\right\}. \]

\section{The second moment}\label{Sec_second}

\noindent
In this section we prove \Prop~\ref{Prop_second}.
We keep the notation and the assumptions of \Sec~\ref{sec:outline} and \Sec~\ref{Sec_first}.

\subsection{Overview}

\noindent
We reduce the problem of estimating $\Erw[\Zkg^2]$ to that of optimising the function $F(a)$
from \Lem~\ref{Lemma_DFGsecond} over a certain domain $\Dtame$.
Due to the additional constraints imposed by the ``tame'' condition, 
this domain $\Dtame$ is a relatively small subset {of $\cD$, which was 
the domain of
optimisation for~(\ref{eqsmmMax}).}
{In the end, $\max_{a\in\Dtame}F(a)$ will be seen to be} significantly smaller than $\max_{a\in\cD}F(a)$, and
additionally, the problem of maximising $F$ over $\Dtame$ technically less demanding.

To define $\Dtame$ formally, call $a\in\cD$ {\em separable} if $a_{ij}\not\in(\cols^{-1}\myconst,\cols^{-1}(1-\kappa))$ for all $i,j\in[\cols]$ (cf.\ (\ref{eqsep})).
Additionally, we say that $a\in\cD$ is {\em $s$-stable} if there are precisely $s$ pairs $(i,j)$ such that $a_{ij}>\cols^{-1}\myconst$. We denote by $\mathcal{D}_s$ the set of all $s$-stable $a\in\mathcal{D}$, and by $\cD_{[\cols-1]}=\cup_{s<\cols}\cD_s$. Geometrically, each $\mathcal{D}_s$ is close to a $(\cols-s)$-dimensional face of the Birkhoff polytope, for if $a$ has entries greater than $(1.01/\unif\cols)^{1/(\unif-1)}$ then by separability these entries are in fact at least $(1-\kappa)/\cols$ (with $\kappa=\ln^{20}\cols/\cols^{\unif-1}$).
Finally, let $\Dtame$ be the (compact) set of all $a\in\cD$ that are separable and $s$-stable for some $0\leq s<\cols$. %\pa{It is worth noting that we maximise over $\mathcal{D}_\text{tame}$ which does not contain $\mathcal{D}_q$, the latter will be dealt with in Lemma \ref{secmomlaplace_2}.}

\begin{lemma}\label{Lemma_Dtame}
 If $F(a)<F(\bar{a})$ for all $a\in\mathcal{D}_{\text{tame}}\backslash\{\bar{a}\}$ then $\Erw[\Zkg^2] {=}  O(\Erw[\Zkb]^2)$.
\end{lemma}

\noindent
The proof of \Lem~\ref{Lemma_Dtame} is by a standard application of the Laplace method.
We defer the details to \Sec~\ref{lapsect}. 

In order to prove that $\max_{a\in\Dtame}F(a)=F(\bar a)$, we observe that the set $\Dtame$ naturally decomposes into a number of disjoint subsets.
Namely, let $\dtame s$ be the set of all $s$-stable $a\in\Dtame$ for $0\leq s<\pa{q}$.
We will argue that for $1\leq s<\cols$ the maximum of $F$ over $\dtame s$ is not much greater than the function value attained at certain canonical points
	$\bar a(s)$ with entries
	\begin{align}\label{baras2}
	\bar a_{ij}(s)= {\cols^{-1}\, \vecone\{i=j\}\vecone\{i\leq s\}} + {(\cols(\cols-s))^{-1}}\,\vecone\{i>s\}\vecone\{j>s\}.
	\end{align}
Hence, $\bar a(s)$ is a block-diagonal matrix.
The upper-left block is the $s\times s$ identity matrix, divided by $\cols$, and the lower-right block is the $(\cols-s)\times(\cols-s)$ matrix with all entries equal to $(\cols(\cols-s))^{-1}$.
Clearly, $\bar a(s)\in\dtame s$.

The following statement, which we prove in \Sec~\ref{Sec_max}, is the heart of the second moment analysis.

\begin{lemma}\label{Lemma_smmTame}
	 We have $F(a)<F(\bar{a})$ for all $a\in\mathcal{D}_{\text{tame}}\backslash\{\bar{a}\}$. 
	%and $F(a)<F(\bar a)$ for all $a\in\mathcal{D}_{s,\text{tame}}$ and
	 
\end{lemma}
\Prop~\ref{Prop_second} follows immediately from \Lem~\ref{Lemma_Dtame} and \Lem~\ref{Lemma_smmTame}.
%\subsubsection{Proof of \pa{\Prop~\ref{Prop_second}}}
%\begin{proof}By \Lem~\ref{Lemma_Dtame} and \Lem~\ref{Lemma_smmTame} it remains only to show that 
%	\begin{equation}\label{eqLemma_as}
%	F(\bar a(s))+\cols^{0.999-\unif}<F(\bar a)\qquad\mbox{ for all }1\leq s<\cols.
%	\end{equation}
%\noindent
%To see that this is the case,
%whence (\ref{sstableentropyenergy}) and~(\ref{sstable}) follow.
%Finally, (\ref{eqLemma_as}) follows from (\ref{baraf}) and~(\ref{sstable}).
%\end{proof}
%\noindent

\subsection{The Laplace method: proof of \Lem~\ref{Lemma_Dtame}}\label{lapsect}
We seek to show that  there exists some positive constant $C(\cols )$  such that 
	\begin{align}\label{sect5secmom}
		\textbf{E}[Z_{\cols ,\text{tame}}^2]
		\leq C(\cols )\cdot \textbf{E}[Z_{\cols ,
		\text{bal}}]^2.
	\end{align}
The expected value of $Z_{\cols,\text{tame}}^2$ can be written as a sum
 over pairs of tame colourings. Define 
 \begin{align*}\mathcal{E}=\{a\in \cR\cap\mathcal{D}_\text{tame}:\norm{a-\bar{a}}_\unif ^\unif <\eta(\cols )\}.
 \end{align*} We split $Z_{\cols,\text{tame}}^2$ into three components as follows:
 \begin{align*}
	 Z_{\cols,\text{tame}}^2=Z_{\cols,\text{tame}}^2\cdot\textbf{1}_\mathcal{E}+Z_{\cols,\text{tame}}^2\cdot\textbf{1}_{\mathcal{D}_{[q-1]}\backslash \mathcal{E}}+Z_{\cols,\text{tame}}^2\cdot\textbf{1}_{\mathcal{D}_q}.
 \end{align*}
First we estimate the contribution of the first summand above by performing a Taylor expansion of
$F$ around $\bar{a}$. 

\begin{lemma}\label{secmomlaplace}
	There exists $C(\cols )$ and $\eta(\cols )$ such that with  we have 
	\begin{align*}
		\textbf{E}[Z_{\cols ,\text{tame}}^2\cdot\textbf{1}_{\mathcal{E}}]\leq C(\cols )\cdot \textbf{E}[Z_{\cols ,\text{bal}}]^2
	\end{align*}
\end{lemma}
\begin{proof}
We may parametrise $\cR\cap\mathcal{D}_{\text{tame}}$ as follows: disregard the  $(\cols ,\cols )$ entry and consider each matrix $a$ as a $\cols ^2-1$ dimensional vector.
Let
	\begin{align*}
		\mathcal{L}:[0,1/\cols ]^{\cols ^2-1}\longrightarrow [0,1/\cols ]^{\cols ^2}, 
			\hspace{1cm}a_{ij}\mapsto  
		\left\{\def\arraystretch{1.2}
  			\begin{array}{@{}c@{\quad}l@{}}
   			 a_{ij} & \text{if $(i,j)\ne (\cols ,\cols )$},\\
  			 1-\sum_{(i,j)\ne (\cols ,\cols )}a_{ij} & \text{otherwise}.\\
  		\end{array}\right.
	\end{align*}
We compute the Hessian of $F\circ \mathcal{L}=H\circ \mathcal{L}+E\circ \mathcal{L}$. For $(i,j)\ne (s,t)$ we have
	\begin{align*}
		\frac{\partial}{\partial a_{ij}}\big(H\circ\mathcal{L}
			(a)\big)\Big\vert_{a=\bar{a}}=0,\hspace{1cm}
		\frac{\partial^2}{\partial a_{ij}^2}\big(H\circ\mathcal{L}
			(a)\big)\Big\vert_{a=\bar{a}}=-2\cols ^2,\hspace{1cm}
		\frac{\partial^2}{\partial a_{ij}\partial a_{\pa{st}}}\big(H\circ\mathcal{L}
			(a)\big)\Big\vert_{a=\bar{a}}=-\cols ^2.
	\end{align*}
Further
	\begin{align*}
		\frac{\partial}{\partial a_{ij}}\norm{\mathcal{L}
			(a)}^\unif _\unif \Big\vert_{a=\bar{a}}=0,\hspace{1cm}
		\frac{\partial^2}{\partial a_{ij}^2}\norm{\mathcal{L}
			(a)}^\unif _\unif \Big\vert_{a=\bar{a}}=\frac{2\unif (\unif -1)}{\cols ^{2\unif -4}},
			\hspace{1cm}
		\frac{\partial^2}{\partial a_{ij}\partial a_{\pa{st}}}\norm{\mathcal{L}
			(a)}^\unif _\unif \Big\vert_{a=\bar{a}}=\frac{\unif (\unif -1)}{\cols ^{2\unif -4}},
	\end{align*}
and so
	\begin{align*}
		\frac{\partial}{\partial a_{ij}}\big(E\circ\mathcal{L}
			(a)\big)&\Big\vert_{a=\bar{a}}=0,\hspace{2cm}
		\frac{\partial^2}{\partial a_{ij}^2}\big(E\circ\mathcal{L}
			(a)\big)\Big\vert_{a=\bar{a}}=\frac{2c\unif (\unif -1)}{\cols ^{2\unif -4}(1-
			\cols ^{1-\unif })^2},
		\\&\hspace{0.3cm}
		\frac{\partial^2}{\partial a_{ij}\partial a_{\pa{st}}}\big(E\circ\mathcal{L}
			(a)\big)\Big\vert_{a=\bar{a}}=\frac{c\unif (\unif -1)}{\cols ^{2\unif -4}(1-\cols ^{1-
			\unif })^2}.
	\end{align*}
Thus, we have that the first derivative of $F\circ\mathcal{L}$ vanishes at $\bar{a}$, and that the Hessian is
	\begin{align*}
	D^2\big(F\circ\mathcal{L}(a)\big)\Big\vert_{a=\bar{a}}=-\cols ^2\left(1-\frac{2c\unif (\unif -1)}{\cols ^{2(\unif -1)}(1-\cols ^{1-\unif })^2}\right)(\text{id} +\textbf{1})
	\end{align*}
where $\textbf{1}$ is the matrix with all all entries equal to one, and $\text{id}$ is the identity matrix. 
As $\text{id}$ is positive definite, $\textbf{1}$ is positive semidefinite and $c<\cols ^{\unif -1}\ln \cols $ 
we have that the Hessian is negative definite at $\bar{a}$. Further, it follows from continuity that there exists some 
$\tilde{\eta}, \tilde{\xi}$ independent of $n$ such that the largest eigenvalue of $D^2\big(F\circ\mathcal{L}\big)$ is 
smaller than $-\tilde{\xi}$ for all points $\norm{a-\tilde{a}}_2<\tilde{\eta}$. Since $\mathcal{L}$ is linear,
there exists some positive $\eta$, independent of $n$, such that for all $a$ such that $\norm{a-\bar{a}}_2<\eta$ we have $\norm{\mathcal{L}^{-1}-\tilde{a}}_2<\tilde{\eta}$. Taylor's theorem then implies that there is some positive $\xi$,
independent of $n$, such that
%% CSG, 16.3.2017, We used to have $(a_{ij} - 1/q)^2$  here and below, but I think it must be $1/q^2$,
%% i.e. the entries of \bar{a}, right??!!  Should we write $q^{-2}$ or $1/q^2$ ???
	\begin{align*}
		F\circ\mathcal{L}(a)\leq F(\bar{a})-\xi \sum_{(i,j)\ne(\cols ,\cols )}(a_{ij}-\cols^{-2} )^2\hspace{1cm}\text{ for all }a:\norm{a-\bar{a}}_2<\eta.
	\end{align*}
% For A5.
As $\mathcal{E}$ satisfies the conditions required for the event $\mathcal{A}$ in 
Lemma~\ref{Lemma_DFGsecond}, we may apply 
(\ref{eqLemma_DFGsecond_precise}) with $\mathcal{A} = \mathcal{E}$ to obtain
\begin{align*}
	\textbf{E}[Z_{\cols ,\text{tame}}^2\cdot\textbf{1}_{\mathcal{E}}]
	&=\exp\left\{{nF(\bar{a})}\right\}\cdot O(n^{(1-\cols ^2)/2})
		\cdot\sum_{a\in 
		\mathcal{E}}\exp\left\{-
		\xi n\sum_{(i,j)\ne(\cols ,\cols )}(a_{ij}- \cols^{-2} )^2\right\}\\
&\leq\exp\left\{{nF(\bar{a})}\right\}\cdot O(n^{(1-\cols ^2)/2})\cdot
		\int_{\mathbb{R}^{\cols ^2-1}}
			\exp\left\{-\xi n\sum_{(i,j)\ne(\cols ,\cols )}
			(z_{ij}-\cols^{-2} )^2\right\}dz_{ij}\\
&\leq\exp\left\{{nF(\bar{a})}\right\}\cdot O(n^{(1-\cols ^2)/2})\cdot
		\left[\int_{\infty}^\infty
			\exp\left\{-\xi n z^2\right\}dz\right]^{\cols ^2-1}
	\leq C(\cols )\cdot \textbf{E}[Z_{\cols ,\text{bal}}]^2
\end{align*}
for some constant $C(\cols)$ depending only on $q$. Here the final inequality follows from (\ref{perfect}).
\end{proof}

There are two remaining cases to consider, namely
 $a\in\mathcal{D}_{[\cols -1]}\backslash \mathcal{E}$ and $a\in\mathcal{D}_{\cols }$. We begin with the latter.
	\begin{lemma}\label{secmomlaplace_2}
		There exists a constant $C(\cols ) > 0$ such that 
	\begin{align*}
		\textbf{E}[Z_{\cols ,
			\text{tame}}^2\cdot\textbf{1}_{\mathcal{D}_{\cols }}]\leq 
			C(\cols )\cdot \textbf{E}[Z_{\cols ,\text{bal}}]^2
	\end{align*}
	\end{lemma}
\begin{proof}
\pa{Recall that $a(\sigma,\tau)\in \mathcal{D}_q$ if and only if  there is a permutation of the colours of $\tau$ such that the resulting colouring is in $\mathcal{C}(\mathcal{G},\sigma)$. Therefore}
	\begin{align*}
		\textbf{E}[Z_{\cols ,\text{tame}}^2\cdot\textbf{1}_{\mathcal{D}_{\cols }}]
		&=\sum_{a(\sigma,\tau) \in\mathcal{D}_\cols }
			\textbf{P}\big[\sigma,\tau\text{ are tame colourings}  \big]
		%=\cols !\sum_{\sigma,\tau \in\mathcal{D}_\cols }
	%		\textbf{P}\big[\pa{\sigma \text{ is a tame colouring}} %\text{ and 		
	%		}\tau\in\mathcal{C}(\hyp,\sigma) \big]\\
		%&=\cols !\sum_{\sigma,\tau \in\mathcal{D}_\cols }
		%	\textbf{P}\big[\tau\in\mathcal{C}(\hyp,\sigma)\vert \text{ 
		%	}\pa{\sigma \text{ is a tame colouring}} \big]\cdot\textbf{P}\big[\pa{\sigma \text{ is a tame colouring}}\big]\\
		\\&\leq\cols !\sum_{\text{balanced}\atop \sigma:[n]\mapsto[q] }
			\textbf{E}\big[\vert\mathcal{C}(\hyp,\sigma)\vert \,\big\vert\,\text{ 
			}\sigma\pa{\text{ is a tame colouring}} \big]\cdot\textbf{P}\big[\sigma	
			\pa{\text{ is a tame colouring}}\big]\\
		&\leq \cols ! \cdot\textbf{E}[Z_{\cols ,\text{bal}}]\sum_{\text{balanced}\atop \sigma:[n]\mapsto[q]  }\textbf{P}\big[\sigma	
			\pa{\text{ is a tame colouring}}\big]
   \leq \cols ! \cdot\textbf{E}[Z_{\cols ,\text{bal}}]^2,\hspace{1.6cm}\text{[by \textbf{T3}]},
	\end{align*}
as desired.
%% In the last line, final = was wrong, replaced by $\leq$, just bounding prob above by 1.
\end{proof}
	\begin{lemma}\label{secmomlaplace_3}
			\pa{If $F(a)<F(\bar{a})$ for all $a\in\mathcal{D}_{\text{tame}}\backslash\{\bar{a}\}$ then} we have  
	\begin{align*}
		\textbf{E}[Z_{\cols ,
			\text{tame}}^2\cdot\textbf{1}_{\mathcal{D}_{[\cols -1]}\backslash 
			\mathcal{E}}]\leq  \textbf{E}[Z_{\cols ,\text{bal}}]^2.
	\end{align*}
	\end{lemma}
\begin{proof}
	We take $\eta$ as in Lemma \ref{secmomlaplace} and set 
	\begin{align*}
		\mathcal{E}'=\{a\in \cR\cap\mathcal{D}_\text{tame}:\norm{a-\bar{a}}_2\geq \eta\}.
	\end{align*}
As $\mathcal{E}'$ is compact, the assumption that $F(a)<F(\bar{a})$ for all $a\in\mathcal{E}'$ additionally implies that there exists some $\gamma$ such that $\max_{a\in\mathcal{E}'}F(a)<F(\bar{a})-\gamma$. Then it follows from Lemma \ref{Lemma_DFGfirst} and (\ref{Fbara}) that
	\begin{align*}
		\textbf{E}[Z_{\cols ,
		\text{tame}}^2&\cdot\textbf{1}_{\mathcal{D}_{[\cols -1]}\backslash 
			\mathcal{E}}]
		\leq |\mathcal{E}'|\exp\left\{n(F(\bar{a})-\gamma)\right\}
		\leq n^{\cols ^2}\exp\left\{n(F(\bar{a})-\gamma)\right\}
		\\&\hspace{2cm}\leq \exp\left\{n(F(\bar{a})-\gamma/2)\right\}
		\leq \textbf{E}[Z_{\cols ,\text{bal}}]^2\cdot\exp\{-n\gamma/3\}\leq 	
		\textbf{E}[Z_{\cols ,\text{bal}}]^2,
	\end{align*}
as desired.
\end{proof}
Finally, (\ref{sect5secmom}) follows from combining Lemmas \ref{secmomlaplace}-\ref{secmomlaplace_3}.

\subsection{The maximisation problem: proof of Lemma \ref{Lemma_smmTame}}\label{Sec_max}

Throughout this subsection it is sufficient to assume that $c$ equals the upper bound of (\ref{c-bounds}), that is,  $$c=(q^{k-1}-1/2)\ln q -\ln 2 -1.01q/\ln q.$$ 
To see this, suppose that Lemma~\ref{Lemma_smmTame} is true with this value of $c$. 
Then $\bar{a}$ is the unique maximum of $F$ on $\Dtame$.
Now
$F$ is the sum of the concave function $H$ and the convex function $E$, which attain their maximum, respectively minimum, at $\bar{a}$. Further, since $H$ is independent of $c$ and $E$ is a linear multiple of $c$, decreasing the value of $c$ only makes the maximum of $F$ at $\bar{a}$ more pronounced.

\subsubsection{The strategy}

The proof is based on the local variation technique developed in \cite{Danny}. Roughly speaking, for each $0<s<\cols$ we will argue that for any arbitary $a\in\mathcal{D}_s$, we  can move slightly toward a nicer matrix while increasing $F$. The new matrix that we produce is then regular enough that we may perform calculations and compare it the point $\bar{a}(s)$ whose first $s$ diagonal entries are $1/\cols$, and whose $(i,j)$-entries are equal to $(\cols(\cols-s))^{-1}$ for $i,j>s$.  As it turns out, $\bar{a}(s)$ comes close enough to maximising $F$ over $\mathcal{D}_s$ (up to a negligible error term in each case). The final step is then to show that $F(\bar{a}(s))$ is strictly less than $F(\bar{a})$.

Let us take a moment to collect some results that will be used throughout the remainder of this section.
In particular, it may come as no surprise that in a local variations argument we make extensive use of derivatives. Taking partials of $F$ we have 
	\begin{align}\label{PartialF_Averaging}
		\bc{\frac{\partial}{\partial a_{ix}}-\frac{\partial}{\partial 
		a_{iy}}}F(a)&=\ln\frac{a_{iy}}{a_{ix}}
		+\frac{c\unif(a_{ix}^{\unif-1}-a_{iy}^{\unif-1})}
		{1-2/\cols^{\unif-1}+\norm{a}_\unif^\unif},\hspace{1cm}i,x,y\in [\cols].
	\end{align}
This represents the change in $F$ when we increase $a_{ix}$ at the expense of $a_{iy}$ (see \pa{\Lem}~\ref{Prop_Averaging}, which describes when the above quantity is positive). Further, we will often tackle the changes in entropy and energy separately. 
{We need the following elementary inequalities (cf.\ \cite[Corollary 4.10]{Danny}).}
As usual, the entropy of a vector $b\in [0,1]^q$ is defined by $H(b) = -\sum_{i\in [q]} b_i\ln b_i$.

\begin{fact}\label{Cor410}
	Let $\cols b\in [0,1]^\cols$ be such that $\sum_{i=1}^\cols b_i=1/\cols$, and define 
\begin{align*}
{	h:[0,1]\rightarrow\mathbb{R},\hspace{1cm}z\mapsto -z\ln z - (1-z)\ln (1-z).}
\end{align*}
Then 
\begin{enumerate}
	\item[(i)]for $J\subseteq [\cols]$ and  $r=\sum_{i\in J}\cols b_i$ we have $H(b) \leq h(r) +r\ln |J|+(1-r)\ln(\cols-|J|)$, and
	\item[(ii)] for $J\subseteq\{2,\dots,\cols\}$ with $0<|J|<\cols-1$ and $r=\sum_{i\in J}\cols b_i$, if $\cols b_1<1$ then
	\begin{align*}
		H(b)\leq h(\cols b_1)+(1-\cols b_1)h(r/(1-\cols a_1))+r\ln |J|+(1-r-\cols a_1)\ln (\cols-|J|-1).
	\end{align*}
\end{enumerate}
\end{fact}

The following lemma is the main tool to carry out the local variations argument.
Recall that $\cS$ is the set of all matrices $a=(a_{ij})_{i,j\in[\cols]}$ with entries $a_{ij}\geq0$ such that $\sum_ja_{ij}=1/\cols$ for all $i$.

\begin{lemma}\label{Prop_Averaging}
Suppose $a\in \cS$.
If $i\in[\cols]$ and $\emptyset\neq J\subseteq [\cols]$ are such that for some number $3\ln \ln \cols/ \ln \cols \leq  \mu \leq 1$ we have
	\begin{align}\label{eqProp_Averaging}
		|J|\geq \cols^\mu\hspace{0.5cm}\text{and}\hspace{0.5cm} 
		\max_{j\in J} 
		a_{ij}^{\unif-1}<\frac{0.995}{\unif\cols^{\unif-1}}\left( \mu - \ln\ln \cols/\ln \cols\right),
	\end{align}
 then the matrix $\tilde{a}\in\cS$ obtained from $a$ by setting
	\begin{align*}
		\tilde{a}_{xy}=\vecone\{(x,y)\notin \{i\} \times J \}a_{xy}+
			\frac{\vecone\{(x,y)\in \{i\} \times J \}}{|J|}\sum_{j\in J}a_{ij}
	\end{align*}
is such that $F(a)\leq F(\tilde{a})$.
In fact, the inequality is strict unless $a=\tilde a$.
\end{lemma}
\begin{proof}
Take $i\in [\cols]$, $J\subset [\cols]$ 
% For C16, typo fixed
as described and $x,y\in J$ such that 
	\begin{align*}
		a_{ix}^{\unif-1}=\min_{j\in J}a^{\unif-1}_{ij}< a^{\unif-1}_{iy}  < \frac{0.995}{\unif
	\cols^{\unif-1}}\left( 		
		\mu - \ln\ln \cols/\ln 
		\cols\right).
	\end{align*}
We will show that (\ref{PartialF_Averaging}) is positive for the range of $a_{ix}$ and $a_{iy}$ that we have at hand. It will be convenient to make the substitution $\delta_{xy}=a_{iy}^{\unif-1}-a_{ix}^{\unif-1}>0$ and instead consider whether 
	\begin{align}\label{Prop_Averaging_PartialEqn}
		(\unif-1)&\bc{\frac{\partial}{\partial a_{ix}}-\frac{\partial}
		{\partial 
		a_{iy}}}F(a)
		=\ln\left(\left(\frac{a_{iy}}{a_{ix}}\right)^{\unif-1}\right)
	-\frac{c\unif(\unif-1)(a_{iy}^{\unif-1}-a_{ix}^{\unif-1})}
	{1-2/\cols^{\unif-1}+\norm{a}_\unif^\unif}\nonumber
		\\&\hspace{5cm}=\ln\left(1+\frac{\delta_{xy}}{a_{ix}^{\unif-1}}\right)
		-\frac{c\unif(\unif-1)\delta_{xy}}{1-2/\cols^{\unif-1}+\norm{a}_\unif^\unif}=:
		\Delta(\delta_{xy})>0.
	\end{align}
After noting that $\Delta(0)=0$, it follows from the concavity of $\Delta$ that if $\delta^{\ast}>0$ satisfies (\ref{Prop_Averaging_PartialEqn}) then so does $\delta_{xy}$ for all $0<\delta_{xy}<\delta^{\ast}$. Therefore we take 
	\begin{align*}
		\delta^\ast =\frac{0.999}{\unif\cols^{\unif-1}}\left( \mu - \ln\ln \cols/\ln 		
			\cols\right)>\max_{x,y\in J}\delta_{xy}=\max_{x,y\in J} \left|
			a_{iy}^{\unif-1}-a_{ix}^{\unif-1}\right|,
	\end{align*}
and observe that $a_{ix}\leq \frac{1}{|J|}\sum_{j\in J} a_{ij}\leq\frac{1}{\cols|J|}$. After taking the exponential of (\ref{Prop_Averaging_PartialEqn}), we have
	\begin{align*}
		&\exp\left\{\frac{c\unif(\unif-1)\delta^\ast}{1-2/\cols^{\unif-1}+\norm{a}_\unif^\unif} 
			\right\}
			<\exp\left\{(\unif-1)\ln \cols\left(\mu - \ln \ln \cols / \ln
			\cols\right)\right\} 
		= \left(\cols^\mu/\ln \cols\right)^{\unif-1}
			\leq \left(|J|/\ln
		 	\cols\right)^{\unif-1}
		 	\\&
		 \hspace{0.5cm}\leq \left(\cols a_{ix}\ln \cols\right)^{1-\unif}
			\leq 1+\frac{1.99\ln\ln \cols}{\unif\cols^{\unif-1}a_{ix}^{\unif-1}\ln \cols} 
		\leq 1+ \frac{0.995}{\unif\cols^{\unif-1}}\left( \mu - \ln\ln 
			\cols/\ln \cols\right)\cdot \frac{1}{a_{ix}^{\unif-1}}
			< 1+\delta^{\ast}/a_{ix}^{\unif-1},
	\end{align*}
	as required.
\end{proof}

\noindent
In other words, if we take a row $i$ and a set $J$ of not too few columns such that the largest entry $a_{ij}$, $j\in J$, is not too big,
then the function value does not drop if we replace all entries $a_{ij}$, $j\in J$, by their average.
Thus, \pa{\Lem}~\ref{Prop_Averaging} can be used to ``flatten'' parts of the matrix $a$ without reducing the function value.

In what follows we will use \pa{\Lem}~\ref{Prop_Averaging} to show
that Lemma \ref{Lemma_smmTame} holds for each $0\leq s <\cols$ separately. Formally, we set out to show that:

\begin{claim}\label{Prop0abar}
For all $a\in\mathcal{D}_{0,\text{tame}}\backslash\{\bar{a}\}$ we have $F(a)<F(\bar{a})$.
\end{claim}
\begin{claim}\label{Prop0999abar}
Suppose that $1\leq s \leq \cols^{0.999}$. Then for all $a\in\mathcal{D}_{s,\text{tame}}$ we have $F(a)<F(\bar{a})$.
\end{claim}
\begin{claim}\label{Prop1049abar}
Suppose that $\cols^{0.999}< s < \cols-\cols^{0.49}$. Then for all $a\in\mathcal{D}_{s,\text{tame}}$ we have $F(a)<F(\bar{a})$.
\end{claim}
\begin{claim}\label{Prop1abar}
Suppose that $\cols-\cols^{0.49}\leq s <\cols$. Then for all $a\in\mathcal{D}_{s,\text{tame}}$ we have $F(a)<F(\bar{a})$.
\end{claim}
\noindent
 \Lem~\ref{Lemma_smmTame} is then immediate from Claims~\ref{Prop0abar}--\ref{Prop1abar}. 
 
 \pa{The general strategy will be to  compare $a\in \mathcal{D}_{s,\text{tame}}$ to the overlap $\bar{a}(s)$ defined in (\ref{baras2}) and finally to the central overlap above.} To this end, we observe that
 \begin{align}\label{baraf}
 H(\bar{a})&=2\ln \cols,
 \hspace{1cm}\text{and}\hspace{1cm}E(\bar{a})=-2\ln \cols+\frac{2\ln 2}{\cols^{\unif-1}}+o(\cols^{1-\unif}),
 \end{align}
 and if $s<q$ then 
 \begin{align}\label{sstableentropyenergy}
 H(\bar{a}(s))&=\frac{s}{\cols }\ln \cols +\frac{\cols -s}{\cols }\ln (\cols (\cols -s)),\qquad
 E(\bar{a}(s))<-2\ln \cols  +\frac{s}{\cols }\ln \cols +\widetilde{O}_\cols (\cols ^{1-\unif })\qquad\mbox{and thus}\\
 F(\bar{a}(s))&=\frac{s}{\cols }\ln \cols +\frac{\cols -s}{\cols }\ln (\cols (\cols -
 s))+c\ln\left(1+\frac{s-2\cols }
 {\cols ^\unif }+\frac{(\cols -s)^2}{\cols ^{\unif }(\cols -s)^\unif } \right)\nonumber\\		
 &<\ln \cols  +\frac{\cols -s}{\cols }\ln(\cols -s)+o_\cols ({\cols ^{1-\unif }})\nonumber\\
 &\hspace{1cm}+\left[ (\cols ^{\unif -1}-1/2)\ln \cols  -\ln 2
 \right]\cdot\left[\left(\frac{s-{2}\cols }
 {\cols ^\unif }+\frac{(\cols -s)^2}{\cols ^{\unif }(\cols -s)^\unif } \right)
 -\frac{1}{2}\left( 
 \frac{s-2\cols }{\cols ^\unif }+\frac{(\cols -s)^2}{\cols ^{\unif }(\cols -s)^\unif } \right)^2 
 \right]\nonumber\\
 &=
 (1-s/\cols )\ln(1-s/\cols )+\frac{2\ln 2}{\cols ^{\unif -1}}-\frac{ s\ln 2}
 {\cols ^\unif } -\frac{s\ln \cols }{2\cols ^{\unif }}+\frac{\ln \cols }{\cols ^{\unif -1}}+\frac{\ln \cols }
 {\cols ^{\unif -1}(1-s/\cols )^{\unif -2}} \nonumber\\
 &\hspace{3.7cm}-
 \frac{\cols ^{\unif -1}\ln \cols }{2}\left[\left( 
 \frac{s-2\cols }{\cols ^\unif }+\frac{(\cols -s)^2}{\cols ^{\unif }(\cols -s)^\unif } \right)^2 
 \right]+o_\cols (\cols ^{1-\unif }).\label{sstable}
 \end{align}

\subsubsection{Proof of Claim \ref{Prop0abar}}
We begin with the following consequence of \pa{\Lem}~\ref{Prop_Averaging}.

\begin{claim}\label{prop_099101}
Suppose that $a\in \cS$ has an entry $a_{ij}\in [1.02/(\cols\unif),\, q^{-1}(1.01/\unif)^{1/(\unif-1)}]$.
Then the matrix $a'\in \cS$ 
with entries
	$$a'_{xy}=\vecone\{x\neq i\}a_{xy}+\vecone\{x=i\}q^{-2}\qquad(x,y\in[\cols])$$
satisfies $F(a')> F(a)$.
\end{claim}
\begin{proof}
Without loss of generality we may assume that $a$ maximises $F(a)$ over the set $a\in\cS$ with respect to $a_{11}\in [1.02/(\cols\unif),\cols^{-1}\myconst]$. If we apply \pa{\Lem}~\ref{Prop_Averaging} to the set $J=[\cols]\backslash \{1\}$ with $\mu =\ln (\cols-1)/\ln \cols$ then the maximality of $F(a)$  implies that
	$a_{1j}=(1-\cols a_{11})/(\cols (\cols -1))$ for $j\geq 2$.
Let $a_1$ denote the first row of $a$.
Because $a'$ is obtained from $a$ by replacing the first row by $(\cols^{-2},\ldots,\cols^{-2})$,
the change in entropy comes to
	\begin{align}\label{eqQuickAndDirty0}
	H(a')-H(a)&=\cols^{-1}\ln \cols-H(\cols a_1)\geq q^{-1}\bc{\ln\cols-\ln2-(1-1.02/\unif)\ln\cols}
		\geq q^{-1}\bc{(1.02\ln q)/\unif-\ln2}.
	\end{align}
Furthermore,
	\begin{align}\label{eqQuickAndDirty1}
	\norm a_\unif^\unif-\|a'\|_\unif^\unif&=
		a_{11}^{\unif}-\cols^{1-2\unif}+(\cols-1)\left[\frac{1-
			\cols a_{11}}{\cols(\cols-1)}\right]^{\unif}\leq
			\cols^{-\unif}\bc{1.01/\unif}^{\unif/(\unif-1)}+4\cols^{1-2\unif}.
	\end{align}
The derivative of the function $E$ from \Lem~\ref{Lemma_DFGsecond} satisfies
	\begin{align}\label{probpartial}
		\frac{\partial E(a)}{\partial \norm{a}_\unif^\unif}=\frac{c}{1-2\cols^{1-\unif}+\norm{a}_\unif^\unif}\leq 1.001\cols^{\unif-1}\ln \cols.
	\end{align}
Hence, (\ref{eqQuickAndDirty1}) implies that
	$E(a)-E(a')\leq
		1.02\unif^{-\unif/(\unif-1)}\cols^{-1}\ln q.$
Combining this bound with (\ref{eqQuickAndDirty0}) and assuming that $\cols\geq q_0$ for a large enough constant $\cols_0$, we find 
	$F(a')-F(a)=H(a')-H(a)+E(a')-E(a)>0$.
\end{proof}

\begin{proof}[Proof of Claim \ref{Prop0abar}.]
The set $\dtame 0$ is compact.
Therefore, the continuous function $F$ attains a maximum at some point $a\in\dtame 0$.
Assume for contradiction that $a\neq\bar a$.
Then we will construct a sequence of matrices $a[i]$, $i\in\brk\cols$, such that $a[0]=a$, $a[\cols]=\bar a$, with
$F(a[i+1])\geq F(a[i])$ for all $i<\cols$ and $F(a[0])\neq F(a[\cols])$, clearly arriving at a contradiction to the maximality of $F(a)$.
Specifically, let $a[0]=a$ and obtain $a[i]$ from $a[i-1]$ by letting
	$$a_{xy}[i]=\vecone\{x\neq i\}a_{xy}[i-1]+\vecone\{x=i\}q^{-2}\qquad\mbox{for }i,x,y\in[\cols].$$
This construction ensures that $a[\cols]=\bar a$.
To show that $F(a[i+1])\geq F(a[i])$ we consider two cases.
\begin{description}
\item[Case 1: $\max_{j\in[\cols]}a_{ij}\leq1.02/(\cols\unif)$.]
	We apply \Lem~\ref{Prop_Averaging} with $J=[\cols]$ and $\mu=1$
	Since $a_{ij}\leq1.02/(\cols\unif)$, the assumption (\ref{eqProp_Averaging}) is satisfied.
	Consequently, $F(a[i])\geq F(a[i-1])$, with equality if and only if $a[i]=a[i-1]$.
\item[Case 2: $\max_{j\in[\cols]}a_{ij}>1.02/(\cols\unif)$.]
	Claim~\ref{prop_099101} shows that $F(a[i])>F(a[i-1])$.
\end{description}
Finally, since $a\neq\bar a$ we have $a[i]\neq a[i-1]$ for some $i\in[\cols]$, whence
$F(\bar a)=F(a[\cols])>F(a[0])=F(a)$. {Note that although we may temporarily leave $\mathcal{D}_{0,\text{tame}}$ during this process, we are guaranteed to return to $\bar{a}\in\mathcal{D}_{0,\text{tame}}$.}
\end{proof}

\bigskip

\subsubsection{Proof of Claim \ref{Prop0999abar}}
The strategy of this proof is to compare an arbitary element of $\mathcal{D}_{s}$ to a matrix that is more evenly distributed (using \pa{\Lem}~\ref{Prop_Averaging}),  which we then compare to the barycentre of the face of $\mathcal{D}$ (i.e.  $\bar{a}(s)$) and finally, to which we compare $\bar{a}$.
Let $1\leq s \leq \cols ^{0.999}$ and take $a\in\mathcal{D}_{s}$. It follows from Corollary \ref{prop_099101} and the definition of separability that we may assume $\cols a_{ii}\geq 1-\kappa$ for $i\leq s$ with $\kappa=\ln^{20} \cols /\cols ^{\unif -1}$, and further, that we may also assume  $\cols a_{ij}<1.02/\unif $ for all $i\ne j \leq s$ and $s<i,j\leq \cols $. Let $\cols \hat{a}$ be the singly-stochastic matrix with entries
\[
  \hat{a}_{ij} = \left\{\def\arraystretch{1.2}
  \begin{array}{@{}c@{\quad}l@{}}
    a_{ij} & \text{if $i\in[\cols ],j\leq s$},\\
   \frac{1}{\cols -s}\sum_{\ell>s}a_{i\ell} & \text{if $i\in[\cols ],j>s$}.\\
  \end{array}\right.
\]
Since $\cols -s=\cols (1-o_\cols (1))$ we may apply \pa{\Lem}~\ref{Prop_Averaging} to $J=[\cols ]\backslash[s]$ for any $i\in[\cols ]$.
	It follows that {$F(a)\leq F(\hat{a})$}. 
	We will now compare $F(\hat{a})$ and $F(\bar{a}(s))$. To this end we must first estimate $F(\hat{a})$. We start with the entropy term. As $\hat{a}$ is stochastic and $\cols \hat{a}_{ii}\geq 1-\kappa$ for $i\leq s$, we find that
	\begin{align*}
		r_i=\cols \sum_{i\ne j}\hat{a}_{ij}=1-\cols {a_{ii}}
			\leq \kappa,\hspace{1cm}\text{for }i\leq s.
	\end{align*}
Further, if we set $r_i=\cols \sum_{j=1}^s\hat{a}$ for $i>s$ then it follows from the fact that $\cols a$ is doubly-stochastic that
	\begin{align*}
		\sum_{i>s}r_i=\cols \sum_{i>s}\sum_{j=1}^s\hat{a}_{ij}= 
		\cols \sum_{i>s}\sum_{j=1}^s{a}_{ij}\leq \kappa s,\hspace{1cm}\text{ for }i>s.
	\end{align*}
Let $\hat{a}_i$ denote the $i$th row of $\hat{a}$.
We know from Fact \ref{Cor410} that
	\begin{align*}
		H(\cols \hat{a}_i)\leq h(r_i)+r_i\ln (\cols -1)\leq h(\kappa )+\kappa \ln \cols \hspace{1cm}\text{for }i\leq s,
	\end{align*}
and
	\begin{align*}
		H(\cols \hat{a}_i)\leq h(r_i) +r_i\ln s +(1-r_i)\ln (\cols -s)\leq h(r_i) +r_i\ln s +\ln (\cols -s),\hspace{1cm}\text{ for }i>s.
	\end{align*}
Since $h$ is concave, it follows that
	\begin{align*}
		\sum_{i>s}H(\cols \hat{a}_i)\leq (\cols -s)\ln(\cols -s)+\sum_{i>s}(h(r_i)-r_i\ln s)\leq  (\cols -s)\ln(\cols -s) + \cols h\left( \frac{\kappa s}{\cols }\right) +\kappa s \ln s.
	\end{align*}
Therefore
	\begin{align}\label{hataentropy}
		H(\hat{a})&=\ln \cols  +\frac{1}{\cols }\sum_{i=1}^{\cols }H(\cols \hat{a}_i)
			\leq \ln \cols +\frac{s}{\cols }\left(h(\kappa) +\kappa \ln \cols \right) 
			+\frac{\cols -s}{\cols }\ln(\cols -s)+
			h\left(\frac{\kappa s}{\cols }\right) +\frac{\kappa s}{\cols } \ln 
			s\nonumber\\
		&\leq \ln \cols  +\frac{\cols -s}{\cols }\ln(\cols -s)+o_\cols (\cols ^{1-
			\unif })\hspace{2.3cm}\text{[as $s\leq \cols ^{0.999}$ and $h(\kappa 
			s/\cols )={\widetilde{O}_\cols} (\cols ^{1-\unif })$]}\nonumber\\
		&=H(\cols \bar{a}(s)) +o_\cols (\cols ^{1-\unif })\hspace{3.95cm}\text{[by 
			(\ref{sstableentropyenergy})]}.
	\end{align}
Next we deal with estimation of the energy term. It will be convenient to break down the problem as follows:
	\begin{align*}
		\norm{\hat{a}}_\unif ^\unif &= \sum_{i\leq s}\sum_{j\leq \cols } 
			\hat{a}_{ij}^\unif +\sum_{i>s}\sum_{j>s} \hat{a}_{ij}^\unif + 
			\sum_{i>s}\sum_{j\leq s} \hat{a}_{ij}^\unif .
	\end{align*}
{As the $\unif $-norm is maximised when summands are as unequal as possible}, we have  ${\norm{\hat{a}_i}_\unif ^\unif }\leq \cols ^{1-\unif }$ for $i\leq s$. Further, by the same logic we have 
% For C17, changed subscript l to \ell
\begin{align*}
\sum_{i>s}\sum_{j>s} \hat{a}_{ij}^\unif 
 = (\cols -s)^2\left(\frac{1}{\cols - s} \sum_{\ell >s}a_{i\ell} \right)^\unif \leq (\cols -s)^{2-\unif }\cols ^{-\unif },
\end{align*}
and 
	\begin{align*}
			\sum_{i>s}\sum_{j\leq s} \hat{a}_{ij}^\unif \leq \left(	\sum_{i>s}\sum_{j\leq s} \hat{a}_{ij}\right)^\unif \leq \left(\frac{\kappa s}{\cols }\right)^\unif .
	\end{align*}
As $s/\cols \leq \cols ^{-0.001}$ we know 
	\begin{align*}
		\sum_{i>s}\sum_{j\leq s} \hat{a}_{ij}^\unif \leq\left(\frac{\kappa s}{\cols }\right)^\unif  \leq \cols ^{\unif (1-\unif )}\cols ^{-\gamma},
	\end{align*}
for some $\gamma>0$. If we combine the above results then we have shown that
	\begin{align*}
		\norm{\hat{a}}^\unif _\unif \leq s\cols ^{1-\unif }+(\cols -s)^{2-\unif }\cols ^{-\unif }+\cols ^{\unif (1-\unif )}\cols ^{-\gamma}=\norm{\bar{a}(s)}_\unif ^\unif +\cols ^{\unif (1-\unif )}\cols ^{-\gamma},
	\end{align*}
and so
	\begin{align}\label{hataprob}
		E(\hat{a})-E(\bar{a}(s))\leq \frac{\partial E(a)}{\partial \norm{a}_\unif ^\unif } \big(\norm{\hat{a}}^\unif _\unif -\norm{\bar{a}(s)}_\unif ^\unif )\leq \cols ^{-(1-\unif) ^2}\cols ^{-\gamma}\ln \cols (1+o_\cols (1/\cols ))=  o_\cols (\cols ^{1-\unif }).
	\end{align}
Therefore it follows from (\ref{hataentropy}) and (\ref{hataprob}) that
	\begin{align*}
		F(a)\leq F(\hat{a})\leq F(\bar{a}(s))+o_\cols (\cols ^{1-\unif }).
	\end{align*} 
{Recalling that $s/\cols \leq \cols ^{-0.001}$, it follows from (\ref{sstable}) that
	\begin{align*}
		F(a)\leq F&(\hat{a})\leq F(\bar{a}(s))+o_\cols (\cols ^{1-\unif })
		\\&= (1-s/\cols )\ln(1-s/\cols )+\frac{2\ln 2}{\cols ^{\unif -1}}-\frac{s\ln 2}
			{\cols ^\unif } -\frac{s\ln \cols }{2\cols ^{\unif }}+\frac{\ln \cols }{\cols ^{\unif -1}}+\frac{\ln \cols }
			{\cols ^{\unif -1}(1-s/\cols )^{\unif -2}} \nonumber\\
			&\hspace{3.7cm}-
			\frac{\cols ^{\unif -1}\ln \cols }{2}\left[\left( 
			\frac{s-2\cols }{\cols ^\unif }+\frac{(\cols -s)^2}{\cols ^{\unif }(\cols -s)^\unif } \right)^2 
			\right]+o_\cols (\cols ^{1-\unif })
		\\&= (1-s/\cols )\ln(1-s/\cols )+\frac{2\ln 2}{\cols ^{\unif -1}}+\frac{\ln \cols }		
			{\cols ^{\unif -1}}+\frac{\ln \cols }{\cols ^{\unif -1}(1-s/\cols )^{\unif -2}} \nonumber\\
			&\hspace{6.5cm}-
			\frac{\cols ^{\unif -1}\ln \cols }{2}\left( 
			\frac{s-2\cols }{\cols ^\unif }\ \right)^2 
			+o_\cols (\cols ^{1-\unif })
		\\&= -\frac{s}{\cols }(1-s/\cols )+\frac{2\ln 2}{\cols ^{\unif -1}}+\frac{\ln \cols }
			{\cols ^{\unif -1}}+\frac{\ln \cols }{\cols ^{\unif -1}(1-s/\cols )^{\unif -2}} 
			-\frac{2\ln \cols }{\cols ^{\unif -1}} 
			+o_\cols (\cols ^{1-\unif })
		\\&\leq (1-s/\cols )\ln(1-s/\cols )+\frac{2\ln 2}{\cols ^{\unif -1}}+o_\cols (\cols ^{1-\unif })=F(\bar{a})-\frac{s}{\cols }(1-s/\cols )+o_\cols (\cols ^{1-\unif }).
	\end{align*}
As the $\frac{s}{\cols }(1-s/\cols )$ is decreasing in $s$, we have shown that $F(a)<F(\bar{a})-1/\cols +1/\cols ^2+o_\cols (\cols ^{1-\unif })$. This implies our original assertion.\qed}

\subsubsection{Proof of Claim \ref{Prop1049abar}}
Let $\cols ^{0.999}<s<\cols -\cols ^{0.49}$ and take $a\in\mathcal{D}_{s}$. As before, we may assume $\cols a_{ii}\geq 1-\kappa$ for $i\leq s$, and $\cols a_{ij}<1.02/\unif $ for all $i\ne j \leq s$ and $s<i,j\leq \cols $. Let $\cols \hat{a}$ be the singly-stochastic matrix with entries 
\[
  \hat{a}_{ij} = \left\{\def\arraystretch{1.2}
  \begin{array}{@{}c@{\quad}l@{}}
    a_{ij} & \text{if $i=j\in [s]$},\\
   \frac{1}{s-1}\sum_{\ell\in [s]\backslash\{i\}}a_{i\ell} &
	 \text{if $i,j\leq s, i\ne j$},\\
   \frac{1}{\cols -s}\sum_{\ell>s}a_{i\ell} & \text{if $j>s$},\\
   \frac{1}{s}\sum_{\ell\leq s}a_{i\ell} & \text{if $j\leq s<i$}.\\
  \end{array}\right.
\]
Since $s,\cols -s>\cols ^{0.49}$ we may apply \pa{\Lem}~\ref{Prop_Averaging} to $J=[\cols ]\backslash[s]$ and $J'=[s]\backslash \{i\}$ for any $i\in[\cols ]$. It follows that $F(a)\leq F(\hat{a})$. To estimate $F(\hat{a})$ we will now define
	\begin{align*}
		r_i=\cols \sum_{j>s}a_{ij}=\cols \sum_{j>s}\hat{a}_{ij}\text{ for }i\leq s,\text{ and }r_i=\cols \sum_{j\leq s}a_{ij}=\cols \sum_{j\leq s}\hat{a}_{ij}\text{ for }i>s,
	\end{align*}
and since $\cols a$ is doubly stochastic,
	\begin{align*}
		r=\sum_{i>s}r_i=\sum_{i\leq s}r_i\leq \sum_{i\leq s}1-\cols a_{ii}\leq \kappa s.
	\end{align*}
Further, we also set
	\begin{align}\label{tidefinition1}
		t_i=\cols \sum_{j\in[s]\backslash\{i\}}\hat{a}_{ij}
		=\cols \sum_{j\in[s]\backslash\{i\}}{a}_{ij}\leq  1-\cols a_{ii}
		\leq \kappa\hspace{1cm}\text{for }i\leq s.
	\end{align}
As before we will now estimate the entropy term and the energy term separately. 
Again, we let $\hat{a}_i$ denote the $i$th row of $\hat{a}$.
We know from Fact \ref{Cor410} (ii) that
	\begin{align*}
		H(\cols \hat{a}_i)&\leq h(\cols a_{ii})+(1-\cols a_{ii})h(t_i/(1-\cols a_{ii}))+t_i\ln 
			(s-1)+(1-t_i-\cols a_{ii})\ln (\cols -s)\\
		&\leq h(\cols a_{ii})+(1-\cols a_{ii})h(r_i/(1-\cols a_{ii}))+t_i\ln 
			s+r_i\ln (\cols -s)\\
		&=-\cols a_{ii}\ln (\cols a_{ii})-t_i\ln t_i-r_i\ln r_i +t_i\ln s+r_i\ln (\cols -s) \leq h(t_i)+t_i \ln s +h(r_i) +r_i\ln (\cols -s), \hspace{1cm}i\leq s,
	\end{align*}
where the last line follows as the function $g:x\mapsto -(1-x)\ln(1-x)$ is decreasing with $g'(x)\leq 1$ for small $x$. If we set $\widetilde{H}=\frac{1}{\cols }\sum_{i\leq s}(h(t_i) +t_i\ln s)$ then it follows from the concavity of $h$ that
	\begin{align*}
		\frac{1}{\cols }\sum_{i\leq s}H(\cols \hat{a}_i)\leq \widetilde{H} +\frac{s}{\cols }h(r/s)+\frac{r}{\cols }\ln(\cols -s).
	\end{align*}
Furthermore,  by Fact \ref{Cor410} (i) and the concavity of $h$, we have 
	\begin{align*}
		\frac{1}{\cols }\sum_{i>s}H(\cols \hat{a_i})\leq \frac{\cols -s}{\cols }h(r/(\cols -
			s))+{\frac r\cols} \ln s+\frac{\cols -s-r}{\cols }\ln (\cols -s).
	\end{align*}
Combining these results, it follows that
	\begin{align*}
		H(\hat{a})\leq \ln \cols  + \widetilde{H} +\left(\frac{s}{\cols }h(r/s)+\frac{r}
			{\cols }\ln(\cols -s)\right)+\left(\frac{\cols -s}{\cols }h(r/(\cols -
			s))+{(r/\cols)} \ln s\right)+\frac{\cols -s-r}{\cols }\ln (\cols -s).
	\end{align*}
Further as $h(x)\leq x(1-\ln x)$, we have
	\begin{align}\label{hataentropy2}
		H(\hat{a})-\widetilde{H}&\leq\ln \cols  +\frac{r}{\cols }\left[2-2\ln r +2\ln s 
			+\ln (\cols -s) \right]+\frac{\cols -s}{\cols }\ln (\cols -s)\nonumber
		\\&\leq \ln \cols  +\frac{r}{\cols }\left(2+3\ln \cols  \right)+\frac{\cols -s}{\cols }\ln 					(\cols -s)+O_\cols (1/\cols ),\hspace{1cm}\text{[as $-z\ln z\leq 1$ for 
			$z\geq 0$]}\nonumber
		\\&=2\ln \cols   +\frac{r}{\cols }\left(2+3\ln \cols  \right)+(1-s/\cols )\ln(1-s/\cols )-\frac{s\ln \cols }{\cols }+O_\cols (1/\cols ).
	\end{align}
Since $s<\cols $, we obtain
	\begin{align}\label{hataentropy2_1}
		\widetilde{H}-\frac{2\ln \cols }{\cols }\sum_{i\leq s}t_i=\frac{1}{\cols }\sum_{i\leq 
			s}\Big(h(t_i)+t_i(\ln s - 2\ln \cols )\Big)\leq \frac{1}{\cols }\sum_{i\leq 
			s}(h(t_i)-t_i\ln \cols )\leq \frac{1}{\cols },
	\end{align}
where the last inequality follows from noting that $\max_{x\in[0,1]}h(x)-x\ln \cols \leq 1/\cols $. Thus, by combining (\ref{hataentropy2}) and (\ref{hataentropy2_1}) we have shown
	\begin{align}\label{prop1049entropyestimate}
		H(\hat{a})\leq 2\ln \cols   +\frac{r}{\cols }\left(2+3\ln \cols  \right)+(1-s/\cols )\ln(1-s/\cols )-\frac{s\ln \cols }{\cols }+\frac{2\ln \cols }{\cols }\sum_{i\leq s}t_i+O_\cols (1/\cols ).
	\end{align}
Next, we move on to estimating the energy term. As before we firstly estimate $\norm{a}_\unif ^\unif $, then we apply a bound for $\partial E / \partial \norm{a}_\unif ^\unif $ in order to approximate $E(\hat{a})-E(\bar{a}(s))$. Firstly note from (\ref{tidefinition1}) that for $i\leq s$, we have
	\begin{align*}
		\hat{a}_{ii}^\unif &\leq \cols ^{-\unif }(1-t_i)^\unif =\frac{1}{\cols ^\unif }-\frac{\unif t_i}{\cols ^\unif }+o_\cols (\cols ^{1-2\unif }),
		\hspace{1cm}\text{and}\\
		\sum_{j\in[s]\backslash\{i\}}\hat{a}_{ij}^\unif &=(s-1)\left(\frac{t_i/\cols }
		{s-1}\right)^\unif \leq \frac{(\kappa/\cols )^{\unif }}{(s-1)^{\unif -1}}\leq (\kappa/\cols )^{\unif }.
	\end{align*}
Moreover, since {$\cols\hat{a}$} is stochastic and $\cols \hat{a}_{ii}\geq 1-\kappa$ if $i\leq s$, we have
	\begin{align*}
		\sum_{j\in [\cols ]\backslash[s]}\hat{a}_{ij}^\unif \leq (\kappa/\cols )^\unif ,\hspace{1cm}\text{for }i\leq s.
	\end{align*}
Combining the above equations yields
	\begin{align*}
		\sum_{i\leq s}\|\hat{a}_i\|_\unif ^\unif \leq s\cols ^{-\unif }-\frac{\unif }
		{\cols ^\unif }\sum_{i\leq 
		s}t_i+o_\cols (\cols ^{1-2\unif }).
	\end{align*}
Since $\cols a_{ii}\geq 1-\kappa$ for $i\leq s$ we have $\cols a_{ij}\leq \kappa$ for $j\leq s<i$. By construction, this implies that $\cols \hat{a}_{ij}\leq \kappa$ for $j\leq s<i$. Furthermore, we have that
	\begin{align*}
		\sum_{i>s}\sum_{j \leq s} \hat{a}_{ij}^\unif \leq \frac{\kappa^\unif s}
		{\cols ^\unif }\hspace{1cm}\text{and}\hspace{1cm}\sum_{i>s}\sum_{j 
		>s}\hat{a}_{ij}=(\cols -s)^2\left(\frac{\sum_{j>s}a_{ij}}{\cols -
		s}\right)^\unif \leq \cols ^{-\unif }(\cols -s)^{2-\unif }.
	\end{align*}
We have shown that
	\begin{align*}
		\norm{\hat{a}}_\unif ^\unif \leq s\cols ^{-\unif }+\cols ^{-\unif }(\cols -s)^{2-\unif }-\frac{\unif }
		{\cols ^\unif }\sum_{i\leq 
		s}t_i +o_\cols (\cols ^{1-2\unif })= \norm{\bar{a}(s)}_\unif ^\unif -\frac{\unif }
		{\cols ^\unif }\sum_{i\leq 
		s}t_i +o_\cols (\cols ^{1-2\unif }),
	\end{align*}
and so from (\ref{sstableentropyenergy}) and (\ref{probpartial}), we have
	\begin{align}\label{prop1049probestimate}
		E(\hat{a})&=E(\bar{a}(s))+\frac{\partial E(a)}
			{\partial\norm{a}_\unif ^\unif }\cdot\big(\norm{\hat{a}}^\unif _\unif -
			\norm{\bar{a}(s)}^\unif _\unif \big)\nonumber
		\\&\leq E(\bar{a}(s)) -\cols ^{\unif -1}\ln \cols \big(1+o_\cols (1/\cols )
			\big)\cdot\left(\frac{\unif }
			{\cols ^\unif }\sum_{i\leq 
			s}t_i +o_\cols (\cols ^{1-2\unif })\right)\nonumber
		\\&\leq E(\bar{a}(s)) -\frac{\unif \ln \cols }
			{\cols }\sum_{i\leq 
			s}t_i +o_\cols (\cols ^{1-\unif })\nonumber
		\\&=-2\ln \cols +\frac{s}{\cols }\ln \cols -\frac{\unif \ln \cols }
			{\cols }\sum_{i\leq	s}t_i+O_\cols (1/\cols ).
	\end{align}
Finally then, it follows from (\ref{prop1049entropyestimate}) and (\ref{prop1049probestimate}) that
	\begin{align*}
		F(a)\leq F&(\hat{a})\leq \frac{r}{\cols }\left(2+3\ln \cols  \right)
			+(1-s/\cols )\ln(1-s/\cols )+\frac{(2-\unif )\ln \cols }	
			{\cols }\sum_{i\leq s}t_i+O_\cols (1/\cols )
			\\&=(1-s/\cols )\ln(1-s/\cols )+O_\cols (1/\cols )\leq -\frac{s}{\cols }(1-
			s/\cols )+O_\cols (1/\cols ).
	\end{align*}
Fortunately, our assumption $\cols ^{0.999}<s<\cols -\cols ^{0.49}$ ensures that $F(a)<0<F(\bar{a})$.
\qed
\subsubsection{Proof of Claim \ref{Prop1abar}}
Let $\cols -\sqrt{\cols }\leq s\leq \cols -1$ and take $a\in\mathcal{D}_{s}$. As before we may assume $\cols a_{ii}\geq 1-\kappa$ for $i\in[s]$, and $\cols a_{ij}<1.02/\unif $ for all $i\ne j \leq s$ and $s<i,j\leq \cols $. 
Let $r_i=\cols \sum_{j\ne i}a_{ij}$. As $\cols a$ is doubly-stochastic and $\cols a_{ii}\geq 1-\kappa$ for $i\leq s$, we have
	\begin{align*}
		r=\sum_{i\leq s}r_i=\cols \sum_{i\leq s}\sum_{j\ne 
		i}a_{ij}=\sum_{i\leq s}1-\cols a_{ii}\leq \kappa s.
	\end{align*}
Further, we let
	\begin{align*}
		t_i=\sum_{j>s}\cols a_{ij},
		\hspace{1cm}\text{and}\hspace{1cm}t=\sum_{i\leq s}t_i.
	\end{align*}
Since $\cols a$ is doubly-stochastic we have
	\begin{align}\label{Prop1tdefn}
		t=\sum_{i\leq s}\sum_{j>s}\cols a_{ij}=\sum_{i>s}\sum_{j\leq s}\cols a_{ij}.
	\end{align}
The strategy of this proof is to compare $F(a)$ to $F(\cols^{-1}\id)$ where $\id$ is the $\cols \times \cols $ identity matrix. To this end we firstly estimate the entropy of $a$. 
Again, let $a_i$ denote the $i$th row of $a$.
We now set {$\overline{H}=\frac{1}{\cols }\sum_{i\leq s}h(\cols a_{ii})$} and as before apply Fact \ref{Cor410} (ii) and the concavity of $h$ to observe that
	\begin{align*}
		\frac{1}{\cols }\sum_{i\leq s}&H(\cols a_i)\leq \frac{1}{\cols }\sum_{i\leq s} 
			h(\cols a_{ii})+r_ih(t_i/r_i)+t_i\ln (\cols -s) + (r_i-t_i)\ln s
		\\&\leq \overline{H}+\frac{r}{\cols }h(t/r)+\frac{t}{\cols }\ln (\cols -s)+\frac{r-
			t}{\cols }\ln s 
		\\&\leq \overline{H}+\frac{t}{\cols }(1-\ln t + \ln r)+\frac{t}{\cols }\ln (\cols -
			s)+\frac{r-t}{\cols }\ln s
			\hspace{1cm}\text{[as $h(z)\leq z(1-\ln z)$].}
	\end{align*}
As $-z\ln z\leq 1$ for $z>0$, we have that $-t\ln t \leq 1$. Furthermore, as $\cols a$ is doubly-stochastic we have that $t\leq \cols -s$, and so
	\begin{align*}
		\frac{t}{\cols }(1-\ln t + \ln r)\leq \frac{\cols -s}{\cols }\cdot(1+\widetilde{O}(1/\cols )\big).
	\end{align*}
Therefore
	\begin{align}\label{Propk_leqs}
		\frac{1}{\cols }\sum_{i\leq s}&H(\cols a_i)\leq\overline{H}+\frac{t}{\cols }\ln (\cols -s)+\frac{r-t}{\cols }\ln s +\frac{\cols -s}{\cols }\big(1+{\widetilde{O}_\cols }(1/\cols )\big).
	\end{align}
We now move to estimating $H(\cols a_i)$ for $i>s$. As is by now routine, we apply Fact \ref{Cor410} (i) along with the concavity of $h$ and (\ref{Prop1tdefn}) to conclude that
	\begin{align}\label{Propk_geqs}
		\frac{1}{\cols }\sum_{i> s}H(\cols a_i)&\leq \frac{1}{\cols }	
			\sum_{i>s}\left[h\left(\sum_{j\leq s}\cols a_{ij}\right)
			+\sum_{j\leq s}\cols a_{ij}\ln 
			(s)+\left(1-\sum_{j\leq s}\cols a_{ij}\right)\ln (\cols -s)\right] \nonumber
		\\&\leq \frac{\cols -s}{\cols }h\left(\frac{t}{\cols -s}\right)+\frac{t}{\cols }\ln s 
			+\frac{\cols -s-t}
			{\cols }\ln(\cols -s)\nonumber
		\\&\leq \frac{\cols -s}{\cols }\ln 2+\frac{t}{\cols }\ln s 
			+\frac{\cols -s-t}
			{\cols }\ln(\cols -s)\hspace{1cm}\text{[as $h(z)\leq \ln 2$ for all $z$]}.
	\end{align}
Finally then, we have from (\ref{Propk_leqs}) and (\ref{Propk_geqs}) that
	\begin{align*}
		H(&a)=\ln \cols  +\overline{H}+\frac{1}{\cols }\sum_{i\leq \cols }H(\cols a_i)
			\\&\leq\ln \cols +\frac{r}
			{\cols }\ln s +\frac{\cols -s}{\cols }\ln 2+\frac{\cols -s}
			{\cols }\ln(\cols -s)+\frac{\cols -s}{\cols }\big(1+\widetilde{O}(1/\cols )\big)
			\\&\leq\ln \cols +\overline{H}+\frac{r}
			{\cols }\ln s +\frac{\cols -s}{\cols }\ln 2+\frac{\cols -s}
			{2\cols }\ln \cols +\frac{\cols -s}{\cols }\big(1+\widetilde{O}(1/\cols )\big),\hspace{1cm}\text{[as $\cols -s\leq \sqrt{\cols }$]}.
	\end{align*} 
Moving on to the energy term, we firstly estimate $\norm{a}_\unif ^\unif $. As the norm is maximised when the summands are widely distributed, we have
	\begin{align*}
		\sum_{i\leq s}\norm{a_i}_\unif ^\unif \leq \kappa^\unif  s+\sum_{i\leq s}a_{ii}^\unif =\sum_{i\leq s}a_{ii}^\unif +o(1/\cols ^{\unif +1}).
	\end{align*}
A similar argument then applies to the remaining $\cols -s$ rows. 
Recalling Corollary \ref{prop_099101}, it follows that
	\begin{align*}
		\sum_{i>s}\norm{a_i}_\unif ^\unif \leq  (\cols -s)\left(\frac{1.02}{\cols \unif }\right)^\unif .
	\end{align*}
Therefore,
	\begin{align*}
		\norm{a}_\unif ^\unif \leq (\cols -s)\left(\frac{1.02}{\cols \unif }\right)^\unif +\sum_{i\leq s}a_{ii}^\unif +o(1/\cols ^{\unif +1}),
	\end{align*}
and so
	\begin{align*}
		\norm{a}_\unif ^\unif -\norm{\cols^{-1}\id }_\unif ^\unif \leq \frac{\cols -s}{\cols ^\unif }\left[\left(\frac{1.02}{\unif }\right)^\unif -1\right]+\frac{1}{\cols ^\unif }\sum_{i\leq 
		s}[(\cols a_{ii})^\unif -1]+o(1/\cols ^{\unif +1}).
	\end{align*}
Finally then, we have from (\ref{probpartial}) that
	\begin{align*}
		E(a)-E(\cols^{-1}\id )&=\frac{\partial E(a)}{\partial a}\big(
			\norm{a}_\unif ^\unif -\norm{\cols^{-1}\id}_\unif ^\unif \big)
		\\&\leq\ln \cols \cdot\frac{\cols -s}{\cols }\left[\left(\frac{1.02}{\unif }\right)^\unif -1\right]+\frac{\ln \cols }{\cols }\sum_{i\leq 
		s}[(\cols a_{ii})^\unif -1])+o_\cols (\ln \cols /\cols ).
	\end{align*}
If we combine our estimates for the entropy and energy terms, we have shown that
	\begin{align*}
		F(a)-F(\cols^{-1}\id )&\leq \overline{H}+\frac{r}
			{\cols }\ln s +\frac{\ln \cols }
			{\cols }\sum_{i\leq 
		s}[(\cols a_{ii})^\unif -1])+o_\cols (\ln \cols /\cols )
		\\&\hspace{6cm}+\ln \cols \cdot\frac{\cols -s}{\cols }\left[\left(\frac{1.02}
			{\unif }\right)^r+\frac{2+\ln 2}{\ln \cols }-1/2\right].
	\end{align*}
Since $\max_{0<z<1}h(z)-z\ln \cols \leq 1/\cols $ and $h(z)=h(1-z)$, if we set $\rho_{ii} =1-\cols a_{ii}$ then 
	\begin{align*}
		 &\overline{H}+\frac{r}{\cols }\ln s +\frac{\ln \cols }{\cols }\sum_{i\leq s}
			[(\cols a_{ii})^\unif -1])\leq\frac{1}{\cols }\sum_{i\leq s}\Big[h(\rho_{ii})
			+\rho_{ii}\ln \cols +\big((1-\rho_{ii})^\unif -1\big)\ln \cols \Big] \\
		\hspace{0.2cm}&\leq \frac{1}{\cols }\sum_{i\leq s}\big[h(\rho_{ii})+\rho_{ii}\ln \cols -\unif \rho_{ii}\ln \cols \big] +O(1/\cols )\leq \frac{1}{\cols }\sum_{i\leq s}\big[h(\rho_{ii})-\rho_{ii}\ln \cols \big] +O(1/\cols )=O(1/\cols ).
	\end{align*} 
Finally then
	\begin{align*}
		F(a)\leq F(\cols^{-1}\id )-\ln \cols  / 3\cols +o_\cols (\ln \cols /\cols )\leq F(\cols^{-1}\id )=\frac{1}{2}F(\bar{a}).
	\end{align*}\qed

\section*{Acknowledgement}
We would like to thank the referees for their many helpful comments which have improved the paper.


\begin{thebibliography}{29}

\bibitem{Barriers}
D.~Achlioptas, A.~Coja-Oghlan:
Algorithmic barriers from phase transitions.
Proc.~49th FOCS (2008) 793--802.


\bibitem{nae}
D.~Achlioptas, C.~Moore:
Random $k$-SAT: two moments suffice to cross a sharp threshold.
SIAM Journal on Computing {\bf 36} (2006) 740--762.


\bibitem{AchMooreHyp2}
D.\ Achlioptas, C.\ Moore: On the 2-colorability of random hypergraphs.
Proc.\ 6th RANDOM (2002) 78--90.

\bibitem{AchNaor}
D.~Achlioptas, A.~Naor:
The two possible values of the chromatic number of a random graph.
Annals of Mathematics {\bf 162} (2005), 1333--1349.

\bibitem{ANP}
D.~Achlioptas, A.~Naor, Y.~Peres:
Rigorous location of phase transitions in hard optimization problems. 
Nature \textbf{435} (2005) 759--764.


\bibitem{AlonKriv}
N.~Alon, M.~Krivelevich: The concentration of the chromatic number of random graphs.
Combinatorica {\bf 17} (1997) 303--313


\bibitem{BBColor}
B.~Bollob\'as: The chromatic number of random graphs.
Combinatorica {\bf8} (1988) 49--55

\bibitem{Covers}
A.~Coja-Oghlan: Upper-bounding the $k$-colorability threshold by counting covers.
Electronic Journal of Combinatorics {\bf 20} (2013) P32.

\bibitem{Danny}
A.~Coja-Oghlan, D.~Vilenchik:
Chasing the $k$-colorability threshold.
Proc.\ 54th FOCS (2013) 380--389.

\bibitem{Lenka}
A.~Coja-Oghlan, L.~Zdeborov\'a:
The condensation transition in random hypergraph 2-coloring.
Proc.~23rd SODA (2012) 241--250.

\bibitem{KostaNAE}
A.~Coja-Oghlan, K.~Panagiotou:
Catching the $k$-NAESAT threshold.
Proc.\ 44th STOC (2012) 899--908.

% For C18
\bibitem{DSS3}
J.~Ding, A.~Sly, N.~Sun: Proof of the satisfiability conjecture for large $k$.
Proc.\ 47th STOC (2015) 59--68.

%\bibitem{DuboisMandler}
%O.\ Dubois, J.\ Mandler: The 3-XORSAT threshold.
%Proc.\ 43rd FOCS (2002) 769--778.

\bibitem{DFG}
M.~Dyer, A.~Frieze, C.~Greenhill:
On the chromatic number of a random hypergraph.
Journal of Combinatorial Theory, Series B.
\textbf{113} (2015), 68-122

\bibitem{ER}
P.\ Erd\H os, A.\ R\'enyi: On the evolution of random graphs.
Magayar Tud.\ Akad.\ Mat.\ Kutato Int.\ Kozl.\ {\bf 5} (1960) 17--61.

%\bibitem{GoerdtFalke}
%A.\ Goerdt, L.\ Falke: Satisfiability Thresholds beyond $k$-XORSAT.
%CSR (2012) 148--159.

\bibitem{FriezeWormald}
A.~Frieze, N.~Wormald:
Random $k$-Sat: a tight threshold for moderately growing $k$.
Combinatorica {\bf 25} (2005) 297--305.

\bibitem{GJR}
C.~Greenhill, S.~Janson, A.~Ruci\' nski:
On the number of perfect matchings in random lifts.
Combinatorics, Probability and Computing, 19 (2010), pp. 791--817.

\bibitem{HatamiMolloy}
H.\ Hatami, M.\ Molloy: Sharp thresholds for constraint satisfaction problems and homomorphisms.
Random Structures and Algorithms {\bf 33} (2008) 310--332.

\bibitem{JLR}
S.~Janson, T.~{\L}uczak, A.~Ruci\'nski: Random Graphs, Wiley, New York,  2000.

\bibitem{Kauzmann48}
W.~Kauzmann: The nature of the glassy state and the behavior of liquids at low temperatures.
Chem.\ Rev.\ {\bf 43} (1948)  219--256.

\bibitem{KrivelevichSudakov}
M.\ Krivelevich, B.\ Sudakov: The chromatic numbers of random hypergraphs.
Random Structures and Algorithms {\bf 12} (1998) 381--403.

\bibitem{pnas}
F.~Krzakala, A.~Montanari, F.~Ricci-Tersenghi, G.~Semerjian, L.~Zdeborova:
Gibbs states and the set of solutions of random constraint satisfaction problems.
Proc.~National Academy of Sciences {\bf104} (2007) 10318--10323.

% For B11
\bibitem{KS}
A.\ Kupavskii, D.\ Shabanov: 
% On $r$-colorability of random hypergraphs.   <-- title of arXiv version!
Colorings of partial Steiner systems and their applications,
Journal of Mathematical Sciences {\bf 206} (2015) 511--538.
%arXiv:1110.1249 (2011).

\bibitem{LuczakColor} 
T.~{\L}uczak: The chromatic number of random graphs.
\COMB\ {\bf11} (1991) 45--54

\bibitem{Luczak}
T.~{\L}uczak: A note on the sharp concentration of the chromatic number of random graphs.
\COMB\ {\bf 11} (1991) 295--297

\bibitem{Matula}
D.\ Matula: Expose-and-merge exploration and the chromatic number of a random graph.
Combinatorica {\bf 7} (1987) 275--284.

% Also updated
%\bibitem{PittelSorkin}
%B.\ Pittel, G.\ Sorkin: The Satisfiability Threshold for k-XORSAT.
%Random Structures \& Algorithms {\bf 25} (2016) 236--268.

\bibitem{ShamirSpencer}
E.~Shamir, J.~Spencer: Sharp concentration of the chromatic number of random graphs $\gnp$.
\COMB\ {\bf 7} (1987) 121--129


\end{thebibliography}
\end{document}